\documentclass[11pt,a4paper]{article}
\usepackage{amssymb}
\usepackage{amsmath}
\usepackage{amsfonts}
\usepackage{bbm}
\usepackage{amsthm}
\usepackage{mathrsfs}
\usepackage{hyperref}
\usepackage{color}
\usepackage[margin=2.41cm]{geometry}
\usepackage[all,cmtip]{xy}
\usepackage[utf8]{inputenc}
\usepackage{graphicx}
\usepackage{varwidth}
\usepackage{comment}

\usepackage{upgreek}
\usepackage{rotating}

\usepackage{tikz}
\usetikzlibrary{shapes.geometric}

\definecolor{darkred}{rgb}{0.8,0.1,0.1}
\hypersetup{
     colorlinks=false,         
     linkcolor=darkred,
     citecolor=blue,
}

\theoremstyle{plain}
\newtheorem{theo}{Theorem}[section]
\newtheorem{lem}[theo]{Lemma}
\newtheorem{propo}[theo]{Proposition}
\newtheorem{cor}[theo]{Corollary}

\theoremstyle{definition}
\newtheorem{defi}[theo]{Definition}

\newtheorem{exx}[theo]{Example}
\newenvironment{ex}
  {\pushQED{\qed}\exx}
  {\popQED\endexx}

\newtheorem{remm}[theo]{Remark}
\newenvironment{rem}
  {\pushQED{\qed}\remm}
  {\popQED\endremm}

\numberwithin{equation}{section}

\def\nn{\nonumber}

\def\bbK{\mathbb{K}}
\def\bbR{\mathbb{R}}
\def\bbC{\mathbb{C}}

\def\bbZ{\mathbb{Z}}

\def\hom{\underline{\mathrm{hom}}}

\def\Imm{\mathrm{Im}}
\def\Ker{\mathrm{Ker}}

\def\id{\mathrm{id}}
\def\supp{\mathrm{supp}}

\def\dd{\mathrm{d}}
\def\vol{\mathrm{vol}}
\def\sc{\mathrm{sc}}
\def\cc{\mathrm{c}}

\def\pc{\mathrm{pc}}
\def\fc{\mathrm{fc}}

\def\1{I}
\def\oone{\mathbbm{1}}
\def\op{\mathrm{op}}

\def\Loc{\mathbf{Loc}}

\def\Vec{\mathbf{Vec}}
\def\Ch{\mathbf{Ch}}

\def\astdgAlg{\mathbf{dg}^{\ast}\mathbf{Alg}}
\def\astdguLie{\mathbf{dg}^{\ast}\mathbf{uLie}}
\def\PCh{\mathbf{PoCh}}

\def\CC{\mathbf{C}}

\def\hAQFT{\mathbf{hAQFT}}

\def\AAA{\mathfrak{A}}
\def\LLL{\mathfrak{L}}

\def\QQQ{\mathfrak{Q}}

\def\FFF{\mathfrak{F}}
\def\Sol{\mathfrak{Sol}}
\def\CCR{\mathfrak{CCR}}

\def\KG{\mathrm{KG}}
\def\YM{\mathrm{YM}}

\def\sk{\vspace{1mm}}

\makeatletter
\let\@fnsymbol\@alph
\makeatother

\title{%
Linear Yang-Mills theory as a homotopy AQFT
}

\author{%
Marco Benini$^{1,2,a}$, 
Simen Bruinsma$^{3,b}$\ and\ 
Alexander Schenkel$^{3,c}$\vspace{4mm}\\
{\small ${}^1$ Fachbereich Mathematik, Universit\"at Hamburg,}\\
{\small Bundesstr.~55, 20146 Hamburg, Germany.}\vspace{2mm}\\
{\small ${}^2$ Dipartimento di Matematica, Universit\`a di Genova,}\\
{\small Via Dodecaneso 35, 16146 Genova, Italy.}\vspace{2mm}\\
{\small ${}^3$ School of Mathematical Sciences, University of Nottingham,}\\
{\small University Park, Nottingham NG7 2RD, United Kingdom.}\vspace{4mm}\\
{\small \begin{tabular}{ll}
Email: & ${}^a$~\texttt{marco.benini@uni-hamburg.de}, \texttt{benini@dima.unige.it}\\
& ${}^b$~\texttt{simen.bruinsma@nottingham.ac.uk}\\
& ${}^c$~\texttt{alexander.schenkel@nottingham.ac.uk}\vspace{2mm}
\end{tabular}
}
}

\date{October 2019}

\begin{document}

\maketitle

\vspace{-5mm}

\begin{abstract}
\noindent It is observed that the shifted Poisson structure (antibracket) on the solution complex of Klein-Gordon and linear Yang-Mills theory on globally hyperbolic Lorentzian manifolds admits retarded/advanced trivializations (analogs of retarded/advanced Green's operators). Quantization of the associated unshifted Poisson structure determines a unique (up to equivalence) homotopy algebraic quantum field theory (AQFT), i.e.\ a functor that assigns differential graded $\ast$-algebras of observables and fulfills homotopical analogs of the AQFT axioms. For Klein-Gordon theory the construction is equivalent to the standard one, while for linear Yang-Mills it is richer and reproduces the BRST/BV field content (gauge fields, ghosts and antifields). 
\end{abstract}

\vspace{-1mm}

\paragraph*{Report no.:} ZMP-HH/19-10, Hamburger Beitr\"age zur Mathematik Nr.\ 789
\vspace{-2mm}

\vspace{-2mm}

\vspace{-1mm}

\renewcommand{\baselinestretch}{0.8}\normalsize
\tableofcontents
\renewcommand{\baselinestretch}{1.0}\normalsize

\newpage


\section{\label{sec:intro}Introduction and summary}
Because of their outstanding significance in physics and their intricate connection
to mathematics, quantum gauge theories continuously attract a high level of 
attention throughout different fields of research. In the context of algebraic 
quantum field theory (AQFT) \cite{HK,BFV}, which is a powerful axiomatic 
framework for quantum field theory on Lorentzian manifolds,
it is a long-standing open problem to identify the characteristic features
of quantum gauge theories and their gauge symmetries from a model-independent
perspective. To support these more abstract developments, concrete
examples of quantum gauge theories were constructed in the context of AQFT. 
Most of these studies focused on the case of Yang-Mills theory with structure group $\bbR$ or $U(1)$,
see e.g.\ \cite{SDH,BDS14,BDHS14,FewsterLang,BeniniMaxwell,BSSdiffcoho},
but there also exist similar developments for e.g.\ linearized gravity \cite{FewsterHunt,BDM,Khavkine,Khavkine2}
and linearized supergravity \cite{HackSchenkel}. In addition to such non-interacting models,
examples of perturbatively interacting quantum gauge theories were constructed
in \cite{HollandsYM,FredenhagenRejzner,FredenhagenRejzner2,TehraniZahn} by
means of an appropriate adaption of the BRST/BV formalism to AQFT.
\sk

One of the main conceptual insights of these studies was the observation
that quantum gauge theories, when formulated traditionally in terms of 
gauge-invariant on-shell observable algebras, are in conflict with crucial axioms of AQFT.
The first observation \cite{Dappiaggi} was that quantum gauge 
theories may violate the isotony axiom of AQFT, which demands that the push-forward
$\AAA(f) : \AAA(M)\to\AAA(N)$ of observables along every spacetime embedding $f:M\to N$
is an injective map. It was later understood that the violation of isotony
is due to topological charges in quantum gauge theories, e.g.\
electric and magnetic fluxes in Abelian Yang-Mills theory, and hence it is a feature
that is expected on physical grounds, see e.g.\ \cite{SDH,BDS14,BDHS14,BeniniMaxwell,BSSdiffcoho,BBSS} 
for a detailed explanation. The second observation is more subtle as it is related
to local-to-global properties (i.e.\ {\em descent}) of AQFTs. Within the traditional formulation
in terms of gauge-invariant on-shell observable algebras, quantum gauge theories have very
poor local-to-global properties as witnessed for example by the observation
in \cite{Dappiaggi,FewsterLang} that Fredenhagen's universal algebra (which is a certain local-to-global construction)
for Abelian Yang-Mills theory fails to encode crucial gauge theoretic features such as 
Dirac's charge quantization and Aharonov-Bohm phases. It was later understood
and emphasized in \cite{BeniniSchenkelSzabo} that the failure of (too naive versions of) 
local-to-global constructions is due to higher categorical structures in classical and 
quantum gauge theories, which are neglected (i.e.\ truncated) when working in a
traditional AQFT setting that is based on gauge-invariant on-shell observables.
\sk

Our approach towards resolving this conflict at the interface of AQFT and gauge theory
is the recent {\em homotopical AQFT program}
\cite{BeniniSchenkelSzabo,fibcat,BeniniSchenkelSchreiber,AQFToperad,involution,hAQFT,BruinsmaSchenkel},
whose aim is to refine the foundations of AQFT by introducing new concepts from higher category theory.
We refer to \cite{BSreview} for a recent summary and state-of-the-art review of this approach.
Informally speaking, the main difference between a {\em homotopy AQFT} and an ordinary AQFT
is that it assigns to each spacetime a higher categorical algebra 
in contrast to an ordinary $\ast$-algebra of observables, such that suitable 
homotopy coherent analogs of the AQFT axioms hold true.
These statements can be made precise by using techniques from operad theory \cite{AQFToperad,involution,hAQFT}.
Such higher observable algebras should be understood as quantizations of function
algebras on higher categorical spaces called  {\em (derived) stacks}, which are crucial for the description
of field and solution spaces in a gauge theory, see e.g.\ \cite{Schreiber} and \cite{BSreview}
for an introduction and also \cite{BeniniSchenkelSchreiber} for a concrete description of the Yang-Mills stack.
In the context of linear and perturbative quantum gauge theory, the higher field and solution spaces
may be described by {\em chain complexes of vector spaces} and the higher
quantum observable algebras by {\em differential graded $\ast$-algebras}.
A more physical approach to such higher categorical structures is given by the BRST/BV formalism,
which has already found many interesting applications in perturbative AQFT, see e.g.\
\cite{HollandsYM,FredenhagenRejzner,FredenhagenRejzner2,TehraniZahn}.
\sk

One of the most pressing current issues of the homotopical AQFT program is that
there is up to now no fully worked out physical example of a quantum gauge theory 
in this framework. (Various oversimplified toy-models appeared 
previously in e.g.\ \cite{fibcat,hAQFT,BruinsmaSchenkel}.) It is the aim of the present paper
to address this issue by constructing a first proper example of a homotopy AQFT, 
namely {\em linear quantum Yang-Mills theory} with structure group $\bbR$ on 
globally hyperbolic Lorentzian manifolds. Let us emphasize that,
even though linear Yang-Mills theory is clearly one of the simplest examples 
of a gauge theory, its construction as a homotopy AQFT is far from trivial
because one has to work consistently within a higher categorical context.
\sk

A central role in our construction is played by (a linear analog of) the 
{\em derived critical locus} of the linear Yang-Mills action functional,
which yields a chain complex that encodes very refined information 
about the solutions to the linear Yang-Mills equation. By general results
of derived algebraic geometry \cite{DAG,DAG2,Pridham}, this chain complex
carries a canonical {\em shifted Poisson structure}, which is the crucial
ingredient in the factorization algebra approach to quantum field theory
by Costello and Gwilliam \cite{CostelloGwilliam}. One of our main 
observations in this paper is that this shifted Poisson structure is trivial 
in homology due to the geometry of globally hyperbolic Lorentzian manifolds 
and that it can be trivialized by two different kinds of homotopies 
that play a similar role as retarded/advanced Green's operators in ordinary field theory.
Taking the difference between a retarded and an advanced trivialization
allows us to define an unshifted Poisson structure and hence to 
study the canonical quantization of linear Yang-Mills theory. One of the 
technical challenges that we address in this paper is a 
homotopical analysis of the construction sketched above, which is
required to ensure that it is meaningful within our higher categorical context,
i.e.\ compatible with quasi-isomorphisms of chain complexes
and chain homotopies between Poisson structures. For this we shall use 
techniques from both model category theory 
\cite{Hovey,Hirschhorn} and homotopical category theory \cite{DHKS,Riehl}.
In order to make the bulk of this paper accessible to a 
broader audience, we limit our use of such homotopical 
techniques to the bare minimum that is required to ensure
consistency of our results.
\sk

Let us now explain in more detail our constructions and results by outlining 
the content of the present paper: In Section \ref{sec:prelim} we recall some
preliminary results concerning retarded/advanced Green's operators
for Green hyperbolic operators on globally hyperbolic
Lorentzian manifolds and concerning chain complexes of vector spaces.
These techniques will be frequently used throughout the whole paper.
In Section \ref{sec:Sol} we introduce a flexible concept of
field complexes for linear gauge theories and compute the solution
complexes corresponding to a quadratic action functional 
via a linear analog of the derived critical locus construction.
We apply these techniques to two explicit examples, given
by Klein-Gordon and linear Yang-Mills theory 
on globally hyperbolic Lorentzian manifolds,
and explain how they relate to the BRST/BV formalism from physics.
In particular, the derived critical locus construction produces
the field content of the BRST/BV formalism, i.e.\ fields,
ghosts and antifields, together with the relevant differentials.
\sk

In Section \ref{sec:Poisson} we describe and analyze the
shifted Poisson structure on the solution complex that 
exists canonically due to its construction as a derived critical
locus. (In the terminology of the BRST/BV formalism, 
this is called the antibracket.) Our main novel observation
is that, for Klein-Gordon and linear Yang-Mills theory 
on globally hyperbolic Lorentzian manifolds,
this shifted Poisson structure is trivial in homology
and that it can be trivialized by two distinct types of homotopies
that play a similar role to retarded/advanced Green's operators
in ordinary field theory. 
We formalize this insight by introducing
an abstract concept of {\em retarded/advanced trivializations} 
(Definition \ref{def:retadvtrivialization}). 
We prove that these trivializations exist for our two running examples
and that they are unique in an appropriate sense: 
For Klein-Gordon theory they are unique, 
while for linear Yang-Mills theory they are
not unique in a strict sense but rather {\em unique up to chain homotopies},
which is an appropriate and expected relaxation within our higher categorical context
of the uniqueness result for retarded/advanced Green's operators in ordinary field theory.
Taking the difference between (a compatible pair of) a retarded and an advanced trivialization
allows us to define an unshifted Poisson structure (Definition \ref{def:unshiftedPoisson}),
which is again unique up to chain homotopies.
\sk

In Section \ref{sec:Quantization} we study in detail homotopical properties of
the canonical commutation relations (CCR) quantization of
unshifted Poisson complexes into differential graded $\ast$-algebras.
Our main result in this section is Proposition \ref{prop:CCRhomotopical},
which proves that CCR quantization is compatible with weak equivalences
of Poisson complexes and also with homotopies of unshifted Poisson structures. 
This proof requires a rather technical result that is 
proven in Appendix \ref{app:technical}. Hence, the examples
obtained by our construction in Section \ref{sec:Poisson} can be quantized
consistently. We spell out the quantization of Klein-Gordon and
linear Yang-Mills theory in this approach explicitly in Examples \ref{ex:quantumKG} and
\ref{ex:quantumYM}.
\sk

In Section \ref{sec:AQFT} we investigate functoriality 
of our constructions and answer affirmatively 
our initial question whether they define examples of 
homotopy AQFTs (Definition \ref{def:hAQFT}).
A key ingredient for these studies is an appropriate concept
of natural retarded/advanced trivializations 
(Definition \ref{def:naturaltrivializations})
and of natural unshifted Poisson structures 
(Definition \ref{def:naturalunshiftedPoisson}).
Our construction in Section \ref{sec:Poisson} 
determines such natural structures for both Klein-Gordon 
and linear Yang-Mills theory, see Proposition \ref{prop:naturalKGandYM}.
The main result of this paper is Theorem \ref{thm:KGandYMhAQFT},
which proves that our construction yields a description of
Klein-Gordon and linear Yang-Mills theory as homotopy AQFTs.
Concerning uniqueness (up to natural weak equivalences) of our
construction via natural retarded/advanced trivializations, 
we observe that there are subtle differences between
Klein-Gordon and linear Yang-Mills theory. While our construction
determines Klein-Gordon theory uniquely (up to natural weak equivalences) 
on the category $\Loc$ of all globally hyperbolic Lorentzian manifolds,
we can currently only ensure uniqueness (up to natural weak equivalences) 
for linear Yang-Mills theory on each slice category
$\Loc/\overline{M}$, for $\overline{M}\in\Loc$. (See 
Theorem \ref{thm:KGandYMhAQFT} and Remarks \ref{rem:hAQFTKG} 
and \ref{rem:hAQFTYM} for the details.) In AQFT terminology,
this means that, even though we successfully provide a construction of
linear Yang-Mills theory as a homotopy AQFT in the locally covariant framework 
\cite{BFV}, we can currently only ensure that each of its restrictions
to a Haag-Kastler style homotopy AQFT on a fixed spacetime $\overline{M}\in\Loc$
is determined uniquely (up to natural weak equivalences) by our methods.
This potential non-uniqueness of linear quantum Yang-Mills theory in the 
locally covariant setting is linked to features of the category of globally hyperbolic spacetimes
$\Loc$ which, in contrast to the slice categories $\Loc/\overline{M}$, has no terminal object.


\section{\label{sec:prelim}Preliminaries}

\subsection{\label{subsec:Green}Green's operators}
We briefly review those aspects of the theory
of Green hyperbolic operators on globally hyperbolic Lorentzian manifolds
that are required for this work. The reader is referred to \cite{BGP,Bar} for the details.
\sk

Let $M$ be an oriented and time-oriented globally hyperbolic 
Lorentzian manifold of dimension $m\geq 2$. Let $F\to M$
be a finite-rank real vector bundle and denote its vector 
space of sections by $\FFF(M) = \Gamma^\infty(M,F)$.
A linear differential operator $P : \FFF(M)\to \FFF(M)$ is called Green hyperbolic
if it admits retarded and advanced Green's operators
$G^\pm : \FFF_\cc(M)\to \FFF(M)$, where the subscript $\cc$ denotes
compactly supported sections. Recall that a retarded/advanced Green's 
operator is by definition a linear map $G^\pm : \FFF_\cc(M)\to \FFF(M)$
which satisfies the following properties:
\begin{itemize}
\item[(i)] $G^\pm P \varphi = \varphi$, for all $\varphi\in\FFF_\cc(M)$;
\item[(ii)] $P G^\pm \varphi = \varphi$, for all $\varphi\in\FFF_\cc(M)$;
\item[(iii)] $\supp(G^\pm\varphi) \subseteq J^\pm_M(\supp(\varphi))$, for all
$\varphi\in\FFF_\cc(M)$, where $J_M^\pm(S)\subseteq M$ denotes the causal 
future/past of a subset $S\subseteq M$.
\end{itemize}
It was proven in \cite{Bar} that retarded/advanced Green's operators are necessarily unique and
that they admit unique extensions
\begin{flalign}
G^\pm\,:\,  \FFF_{\pc/\fc}(M)~\longrightarrow ~\FFF_{\pc/\fc}(M)
\end{flalign}
to sections with past/future compact support, such that the three properties above 
hold true for all $\varphi\in \FFF_{\pc/\fc}(M)$.
(Recall that $s\in\FFF(M)$ has past/future compact support if there exists
a Cauchy surface $\Sigma\subset M$ such that $\supp(s) \subseteq J^\pm_M(\Sigma)$.)
The difference $G := G^+ - G^- : \FFF_\cc(M)\to \FFF(M)$ of the retarded and 
advanced Green's operator (on compactly supported sections) is often called
the causal propagator. From the properties of $G^\pm$ it follows that
\begin{flalign}
\xymatrix{
0 \ar[r] & \FFF_\cc(M) \ar[r]^-{P} & \FFF_\cc(M) \ar[r]^-{G} & \FFF_\sc(M) \ar[r]^-{P} & \FFF_\sc(M) \ar[r] & 0
}
\end{flalign}
is an exact sequence, where the subscript $\sc$ denotes sections
of spacelike compact support. (Recall that $s\in\FFF(M)$ has spacelike compact support if there exists
a compact subset $K\subseteq M$ such that $\supp(s) \subseteq J^+_M(K)\cup J^-_M(K)$.)
In particular, this implies that $P G = 0 = GP$.
\sk

For every vector bundle $F\to M$ that is endowed with a fiber metric $h$
one can define the integration pairing
\begin{flalign}\label{eqn:fiberpairingintegration}
\langle s,s^\prime\rangle \,:=\, \int_M h(s,s^\prime)\,\vol_M\quad, 
\end{flalign}
for all $s,s^\prime\in \FFF(M)$ with compactly overlapping support. 
Let us consider two such vector bundles $F_1\to M$ and $F_2\to M$
with fiber metrics and a linear differential operator $Q : \FFF_1(M)\to \FFF_2(M)$.
There exists a formal adjoint differential operator
$Q^\ast : \FFF_2(M)\to \FFF_1(M)$ defined by 
\begin{flalign}
\langle s_2 , Qs_1\rangle_2^{} \,=\, \langle Q^\ast s_2,s_1\rangle_{1}^{}\quad,  
\end{flalign}
for all $s_1\in \FFF_1(M)$ and $s_2 \in \FFF_2(M)$ with compactly overlapping support.
A linear differential operator $P : \FFF(M)\to \FFF(M)$ with source and target
determined by the same vector bundle $F\to M$ with fiber metric $h$ 
is called formally self-adjoint if $P^\ast = P$. If $P : \FFF(M)\to \FFF(M)$ 
is a formally self-adjoint Green hyperbolic operator, then its Green's operators
satisfy
\begin{flalign}\label{eqn:Gpmadjoint}
\langle \varphi, G^+ \varphi^\prime \rangle \,=\, \langle G^- \varphi,\varphi^\prime\rangle\quad, 
\end{flalign}
for all $\varphi,\varphi^\prime \in \FFF_\cc(M)$. This implies that the 
causal propagator $G = G^+ - G^-$ is formally skew-adjoint, i.e.\
\begin{flalign}\label{eqn:Gskewadjoint}
\langle \varphi, G \varphi^\prime \rangle \,=\, - \langle G \varphi,\varphi^\prime\rangle\quad, 
\end{flalign}
for all $\varphi,\varphi^\prime \in \FFF_\cc(M)$.
\begin{ex}\label{ex:forms}
The following class of examples is most relevant for this work.
Consider the $p$-th exterior power $F=\bigwedge^p T^\ast M \to M$ of the cotangent bundle.
Then $\FFF(M) = \Omega^p(M)$ is the vector space of $p$-forms. The orientation and Lorentzian
metric on $M$ define a Hodge operator $\ast : \Omega^p(M)\to \Omega^{m-p}(M)$
and thereby a fiber metric, whose integration pairing reads as
\begin{flalign}
\langle \omega,\zeta\rangle = \int_M \omega\wedge\ast\zeta\quad,
\end{flalign}
for all $\omega,\zeta\in \Omega^p(M)$ with compactly overlapping support.
The de Rham differential $\dd : \Omega^p(M)\to \Omega^{p+1}(M)$ is a linear differential
operator and its formal adjoint is the codifferential 
$\delta:= \dd^\ast : \Omega^{p+1}(M)\to \Omega^{p}(M)$.
The d'Alembert operator on $p$-forms is defined by 
\begin{flalign}\label{eqn:dAlembert}
\square :=\delta\dd + \dd\delta \,:\, \Omega^p(M)\longrightarrow \Omega^p(M)
\end{flalign}
and it is formally self-adjoint. Because of $\dd^2 =0$ 
and $\delta^2 =0$, the d'Alembert operators in different degrees are related by
\begin{flalign}\label{eqn:squareddelta}
\dd\,\square = \square\, \dd \quad,\qquad \delta\,\square = \square\,\delta\quad.
\end{flalign}
The d'Alembert operators are Green hyperbolic and because of 
\eqref{eqn:squareddelta} the retarded/advanced Green's operators 
in different degrees are related by
\begin{flalign}
\dd\, G^\pm = G^\pm \,\dd \quad,\qquad\delta\,G^\pm = G^\pm\,\delta\quad.
\end{flalign}
Finally, we note that the Klein-Gordon-type operators $ \square -m^2 : \Omega^p(M)\to \Omega^p(M)$,
where $m \in \bbR_{\geq0}$ is a mass term, are formally self-adjoint Green hyperbolic operators too.
\end{ex}

\subsection{\label{subsec:Ch}Chain complexes}
Chain complexes play a crucial role in formulating and 
proving our results in this paper. The present subsection
contains a brief review of basic aspects of the theory 
of chain complexes that are necessary for this work. This will in 
particular allow us to fix the notations and conventions that we employ
in the main part of this paper. For more details on chain complexes
we refer to \cite{Weibel} and also to \cite{Hovey}.
\sk

Let us fix a field $\bbK$ of characteristic zero and consider $\bbK$-vector spaces. 
In the main sections $\bbK$ will be either the real numbers $\bbR$ or the complex numbers $\bbC$.
A chain complex is a family of vector spaces 
$\{V_n\}_{n\in\bbZ}^{}$ together with a differential, i.e.\ a family of linear maps 
$\{\dd_n : V_n\to V_{n-1}\}_{n\in\bbZ}^{}$ such that $\dd_{n-1}\, \dd_n =0$ for all $n\in\bbZ$. 
To simplify notations, we often denote this data collectively by $V$ and write $\dd : V_n\to V_{n-1}$
for every component of the differential. A chain map $f : V\to W$ is a family of linear maps 
$\{f_n : V_n\to W_n \}_{n\in\bbZ}^{}$ that is compatible with the differentials, i.e.\ $\dd\,f_n = f_{n-1}\,\dd$ 
for all $n\in\bbZ$. We denote by $\Ch_\bbK$ the category of chain complexes of 
$\bbK$-vector spaces with chain maps as morphisms.
\sk

The tensor product $V\otimes W\in\Ch_\bbK$ of two chain complexes 
$V,W\in\Ch_\bbK$ is defined by
\begin{flalign}\label{eqn:tensorproduct}
(V\otimes W)_n\,:=\,\bigoplus_{m\in\bbZ} V_m\otimes W_{n-m}\quad,
\end{flalign}
for all $n\in\bbZ$, together with the differential obtained by the graded Leibniz rule 
$\dd(v\otimes w) := \dd v \otimes w + (-1)^{m}\,  v\otimes \dd w$, for all $v\in V_m$ and $w\in W_{n-m}$.
Note that the $\otimes$ on the right-hand side of \eqref{eqn:tensorproduct} is the 
tensor product of vector spaces. The unit for this tensor product
is given by $\bbK\in\Ch_\bbK$, which we regard as a chain complex concentrated in degree $0$
with trivial differential. The tensor product of chain complexes is symmetric via the chain
isomorphisms $\gamma : V\otimes W\to W\otimes V$ defined by the usual sign-rule
$\gamma(v\otimes w) := (-1)^{m\,k} \,w\otimes v$, for all $v\in V_m$ and $w\in W_k$.
Finally, the mapping complex $\hom(V,W)\in\Ch_\bbK$ between two chain complexes
$V,W\in\Ch_\bbK$ is defined by
\begin{subequations}\label{eqn:internalhom}
\begin{flalign}
\hom(V,W)_n \,:=\, \prod_{m\in\bbZ} \mathrm{Lin}(V_m, W_{n+m})\quad,
\end{flalign}
for all $n\in\bbZ$, where $ \mathrm{Lin}$ denotes the vector space of linear maps
between vector spaces, together with the ``adjoint'' differential 
$\partial : \hom(V,W)_n\to \hom(V,W)_{n-1}$  defined by
\begin{flalign}
\partial L \,:=\, \big\{\dd\,L_m - (-1)^{n}\, L_{m-1}\, \dd: V_m\to W_{n-1+m}\big\}_{m\in\bbZ}^{} \in \hom(V,W)_{n-1}\quad,
\end{flalign}
\end{subequations}
for all $L = \{L_m : V_m\to W_{n+m} \}_{m\in\bbZ}^{} \in \hom(V,W)_n$.
In summary, this endows $\Ch_\bbK$ with the structure of a closed symmetric monoidal category.
\sk

To every chain complex $V\in\Ch_\bbK$ one can assign its homology
$H_\bullet(V) = \{H_n(V)\}_{n\in\bbZ}^{}$, which is the graded vector space
defined by $H_n(V) := \Ker(\dd : V_n\to V_{n-1})/\Imm(\dd : V_{n+1}\to V_n)$, for all $n\in\bbZ$.
A chain map $f:V\to W$ is called a quasi-isomorphism if it induces an isomorphism
$H_\bullet(f) : H_\bullet(V)\to H_\bullet(W)$ in homology. Quasi-isomorphic
chain complexes should be regarded as ``being the same'', which can be made precise
by using techniques from model category theory \cite{Hovey} or 
$\infty$-category theory \cite{HTT,HA}. It is proven in \cite{Hovey} that
$\Ch_\bbK$ carries the structure of a symmetric monoidal model category, 
whose weak equivalences are the quasi-isomorphisms and fibrations are the degree-wise surjective maps.
Every object in the model category  $\Ch_\bbK$ is both fibrant and cofibrant. 
Readers who are not familiar with model categories should 
read the previous statements informally as that ``there exists technology to perform
a variety of constructions with chain complexes that are compatible with quasi-isomorphisms''.
In this paper we try to keep the model categorical technicalities to a bare minimum. We refer
to \cite{BSreview} for a detailed explanation why such 
techniques are conceptually crucial for formalizing (quantum) gauge theories.
\sk

Let us also briefly recall the concept of chain homotopies. 
A chain homotopy between two chain maps $f,g : V\to W$
is a family of linear maps $\lambda = \{\lambda_n : V_n \to W_{n+1}\}_{n\in\bbZ}^{}$
such that $f_n - g_n = \dd\,\lambda_n + \lambda_{n-1}\,\dd$, for all $n\in\bbZ$.
This definition can be rephrased very conveniently by using the mapping complexes
from \eqref{eqn:internalhom}. Note that a chain map $f:V\to W$ is precisely
a $0$-cycle in $\hom(V,W)\in\Ch_\bbK$, i.e.\ an element $f\in\hom(V,W)_0$
of degree $0$ satisfying $\partial f =0$. A chain homotopy 
between two chain maps $f,g : V\to W$ is precisely a $1$-chain in $\hom(V,W)\in\Ch_\bbK$,
i.e.\ an element $\lambda\in \hom(V,W)_1$ of degree $1$, such that $\partial\lambda = f-g$.
Observe that such chain homotopies exist if and only if the homology class
$[f-g]\in H_0(\hom(V,W))$ vanishes. This picture immediately generalizes to higher homotopies:
Given two chain homotopies $\lambda,\lambda^\prime \in \hom(V,W)_1$ between
$f,g:V\to W$, then $\lambda-\lambda^\prime$ is a $1$-cycle in $\hom(V,W)\in\Ch_\bbK$, i.e.\
$\partial(\lambda-\lambda^\prime) =0$. A (higher) chain homotopy between 
$\lambda$ and $\lambda^\prime$ is a $2$-chain $\chi\in \hom(V,W)_2$
such that $\partial\chi = \lambda-\lambda^\prime$. Observe that such (higher) chain homotopies
exist if and only if the homology class $[\lambda-\lambda^\prime]\in H_1(\hom(V,W))$ vanishes. 
The pattern for even higher chain homotopies is now evident.
\sk

We conclude this subsection by fixing our conventions for shiftings (also called suspensions)
of chain complexes. Given any $V\in \Ch_\bbK$ and $p\in\bbZ$, we define $V[p]\in \Ch_\bbK$
by $V[p]_n := V_{n-p}$, for all $n\in\bbZ$, together with the differential 
$\dd_{n}^{V[p]} := (-1)^p\,\dd_{n-p}^V$,
where we temporarily used a superscript on $\dd$ in order to indicate the relevant chain complex.
Note that $V[p][q] = V[p+q]$, for all $p,q\in\bbZ$, and that $V[0] = V$. From the definition of
the tensor product \eqref{eqn:tensorproduct}, one finds that $V[p] \cong \bbK[p]\otimes V$.
For every $V,W\in\Ch_\bbK$ and $p\in\bbZ$, there exists a chain isomorphism 
$\hom(V,W[p]) \cong \hom(V,W)[p]$ determined by the components
\begin{flalign}
\nn \hom(V,W[p])_n ~&\longrightarrow~ \hom(V,W)[p]_n\quad,\\
\{L_m : V_m\to W[p]_{n+m}\}_{m\in\bbZ}^{~}~&\longmapsto~\{L_m: V_m\to W_{n-p+m}\}_{m\in\bbZ}^{~}\quad. \label{eqn:shiftmappingcomplex}
\end{flalign}


\section{\label{sec:Sol}Field and solution complexes}
Let $M$ be an oriented and time-oriented globally hyperbolic
Lorentzian manifold of dimension $m\geq 2$. In this section
all chain complexes will be over $\bbR$, i.e.\  the relevant category is $\Ch_\bbR$.
Our aim is to investigate the solution chain complexes 
for a class of linear gauge field theories on $M$, which we will obtain
from a derived critical locus construction. The following definition
will be self-explanatory after Examples \ref{ex:fieldKG} and \ref{ex:fieldYM}.
\begin{defi}\label{def:fieldcomplex}
A {\em field complex} on $M$ is a chain complex
\begin{flalign}\label{eqn:fieldcomplex}
\FFF(M)\,:=\, \Big(
\xymatrix@C=2em{
\FFF_{0}(M) & \ar[l]_-{Q} \FFF_1(M)
}
\Big)
\end{flalign}
concentrated in homological degrees $0$ and $1$, where 
\begin{itemize}
\item[(i)] $\FFF_n(M) = \Gamma^\infty(M,F_n)$ is the vector space of sections of
a finite-rank real vector bundle $F_n\to M$ with fiber metric $h_n$, for $n=0,1$, and

\item[(ii)] $Q : \FFF_1(M)\to \FFF_0(M)$ is a linear differential operator.
\end{itemize}
\end{defi}

\begin{ex}\label{ex:fieldKG}
Scalar fields on $M$ are described by the field complex
\begin{flalign}
\Big(
\xymatrix@C=2em{
\Omega^0(M) & \ar[l]_-{0} 0
}
\Big)
\end{flalign}
concentrated in homological degree $0$. The fiber metrics are the ones obtained 
from the Hodge operator, see Example \ref{ex:forms}. 
The elements in degree $0$ are interpreted as scalar fields $\Phi\in\Omega^0(M)$
and triviality of the complex in degree $1$ means that there are no gauge transformations,
as it should be in a scalar field theory.
\end{ex}

\begin{ex}\label{ex:fieldYM}
Gauge fields with structure group $G=\bbR$ on $M$ are described by the field complex
\begin{flalign}
\Big(
\xymatrix@C=2em{
\Omega^1(M) & \ar[l]_-{\dd} \Omega^0(M)
}
\Big)
\end{flalign}
where $\dd$ is the de Rham differential. The fiber metrics are the ones obtained from the Hodge operator, see 
Example \ref{ex:forms}.  The elements in degree $0$ are interpreted as
gauge fields $A\in \Omega^1(M)$ and the elements in degree $1$ as 
gauge transformations $\epsilon\in\Omega^0(M)$. The differential $\dd$
encodes how gauge transformations act on gauge fields, i.e.\ $A\to A +\dd\epsilon$.
\end{ex}

\begin{rem}
We would like to mention very briefly that Definition \ref{def:fieldcomplex}
admits an obvious generalization to longer complexes
\begin{flalign}
\FFF(M)\,=\, \Big(
\xymatrix@C=2em{
\FFF_{0}(M) & \ar[l]_-{Q_1} \FFF_1(M)& \ar[l]_-{Q_2} \FFF_2(M) & \ar[l]_-{Q_3} \cdots
}
\Big)\quad,
\end{flalign}
where each $\FFF_n(M) = \Gamma^\infty(M,F_n)$ is the vector space of sections of
a finite-rank real vector bundle $F_n\to M$ with fiber metric $h_n$ and each
$Q_n : \FFF_n(M)\to \FFF_{n-1}(M)$ is a linear differential operator. Such generalization
is relevant for the description of {\em higher} gauge theories, which include 
gauge transformations between gauge transformations. For example, the complex
\begin{flalign}
\Big(
\xymatrix@C=2em{
\Omega^p(M) & \ar[l]_-{\dd} \Omega^{p-1}(M)& \ar[l]_-{\dd} \cdots & \ar[l]_-{\dd} \Omega^0(M)
}
\Big)
\end{flalign}
describes $p$-form gauge fields $A\in\Omega^p(M)$ with gauge transformations
$A\to A+\dd\Lambda$, for $\Lambda\in\Omega^{p-1}(M)$,
$2$-gauge transformations $\Lambda \to \Lambda + \dd\lambda$, for $\lambda\in\Omega^{p-2}(M)$,
and so on. Our results and constructions in this paper apply to this more general case as well, however
we decided to focus on $1$-gauge theories as in Definition \ref{def:fieldcomplex} in order to
improve readability. In particular, our main examples of interest
are described by $2$-term field complexes, see Examples \ref{ex:fieldKG} and \ref{ex:fieldYM}.
\end{rem}

In order to encode the dynamics, we consider a formally self-adjoint 
linear differential operator
\begin{flalign}
P \,:\, \FFF_0(M)~\longrightarrow~ \FFF_0(M)\quad,
\end{flalign}
which we interpret as the equation of motion operator for the fields of the theory.
The corresponding quadratic action functional
\begin{flalign}\label{eqn:action}
S(s_0) \,:=\,\frac{1}{2} \,\langle s_0 , P s_0\rangle  = \frac{1}{2}\, \int_M h_0(s_0, P s_0)\,\vol_M
\end{flalign}
is given by the integration pairing \eqref{eqn:fiberpairingintegration}.
This action is gauge-invariant if and only if $P$ satisfies
\begin{subequations}\label{eqn:PQzero}
\begin{flalign}
P\,Q \,=\, 0\quad,
\end{flalign}
which from now on is always assumed.
Because $P$ is formally self-adjoint, it follows that
\begin{flalign}
0 = (P\,Q)^\ast = Q^\ast\,P^\ast = Q^\ast\,P\quad.
\end{flalign}
\end{subequations}
The variation of the action defines a section $\delta^{\mathrm{v}}S : \FFF(M)\to T^\ast\FFF(M)$
of the cotangent  bundle over $\FFF(M)$. As in \cite[Section 3.4]{BSreview}, we 
define the latter as the product complex
\begin{subequations}
\begin{flalign}
T^\ast \FFF(M) \,:=\, \FFF(M)\times \FFF_\cc(M)^\ast
\end{flalign}
with
\begin{flalign}
\FFF_\cc(M)^\ast \,:=\, \Big(
\xymatrix@C=2em{
\stackrel{(-1)}{\FFF_{1}(M)} & \ar[l]_-{-Q^\ast} \stackrel{(0)}{\FFF_0(M)}
}
\Big)
\end{flalign}
the smooth dual of the compactly supported
field complex $\FFF_\cc(M)$. Here and in the following we use round brackets to
indicate homological degrees. Explicitly, we obtain
\begin{flalign}
T^\ast\FFF(M) \,=\, \Big(
\xymatrix@C=2.5em{
\stackrel{(-1)}{\FFF_1(M)} &\ar[l]_-{-Q^\ast \pi_2}\stackrel{(0)}{\FFF_0(M)\times \FFF_0(M)} & \ar[l]_-{\iota_1 Q} \stackrel{(1)}{\FFF_1(M)}
}
\Big)\quad,
\end{flalign}
\end{subequations}
where $\iota_1 : \FFF_0(M)\to \FFF_0(M)\oplus \FFF_0(M)= 
\FFF_0(M)\times \FFF_0(M)$  denotes the inclusion into the first factor and 
$\pi_2 :  \FFF_0(M)\times \FFF_0(M) \to  \FFF_0(M)$ the projection onto the second factor.
The chain map $\delta^\mathrm{v}S: \FFF(M)\to T^\ast \FFF(M)$ obtained by varying the action
then reads explicitly as
\begin{flalign}
\parbox{0.5cm}{\xymatrix{
\FFF(M)\ar[d]_-{\delta^\mathrm{v}S}\\
T^\ast\FFF(M)
}
}~~=~~~~ \left(\parbox{2cm}{\xymatrix@C=2em{
\ar[d]_-{0}0  &\ar[l]_-{0} \ar[d]_-{(\id, P )}\FFF_0(M) & \ar[l]_-{Q} \FFF_1(M)\ar[d]_-{\id} \\
\FFF_1(M) &\ar[l]^-{-Q^\ast \pi_2} \FFF_0(M)\times\FFF_0(M) & \ar[l]^-{\iota_1 Q} \FFF_1(M)
}}\right)\quad.
\end{flalign}
Note the appearance of the equation of motion operator $P : \FFF_0(M)\to \FFF_0(M)$
in the middle vertical arrow. Hence, in order to enforce the equation of motion, 
we have to intersect $\delta^\mathrm{v}S$ with the zero-section 
\begin{flalign}
\parbox{0.5cm}{\xymatrix{
\FFF(M)\ar[d]_-{0}\\
T^\ast\FFF(M)
}
}~~=~~~~~ \left(\parbox{2cm}{\xymatrix@C=2em{
\ar[d]_-{0}0  &\ar[l]_-{0} \ar[d]_-{(\id, 0 )}\FFF_0(M) & \ar[l]_-{Q} \FFF_1(M)\ar[d]_-{\id} \\
\FFF_1(M) &\ar[l]^-{-Q^\ast \pi_2} \FFF_0(M)\times\FFF_0(M) & \ar[l]^-{\iota_1 Q} \FFF_1(M)
}}\right)\quad.
\end{flalign}
This is the content of the following
\begin{defi}\label{def:solutioncomplex}
Let $\FFF(M)$ be a field complex on $M$ 
and $P : \FFF_0(M)\to\FFF_0(M)$ a formally self-adjoint linear differential 
operator satisfying \eqref{eqn:PQzero}.
The corresponding {\em solution complex} on $M$ is defined as the 
derived critical locus of the action functional $S$ in \eqref{eqn:action}.   
Concretely, it is given by the homotopy pullback
\begin{flalign}\label{eqn:derivedlocus}
\xymatrix{
\ar@{-->}[d]\Sol(M)\ar@{-->}[r] & \ar@{}[dl]_-{h~~~~~} \FFF(M)\ar[d]^-{\delta^{\mathrm{v}}S}\\
\FFF(M)\ar[r]_-{0} & T^\ast \FFF(M)
}
\end{flalign}
in the model category $\Ch_\bbR$.
\end{defi}
\begin{rem}
We would like to add an informal discussion of the important role of homotopy pullbacks 
(see e.g.\  \cite{Hovey,Hirschhorn}) for the benefit of those readers who are not 
familiar with model categories. First, let us note that if \eqref{eqn:derivedlocus} 
would be an ordinary categorical pullback, then it would enforce the equation
of motion in a strict fashion, i.e.\ $P s_0 = 0$. There are however
problems with this naive approach, because it is not guaranteed that replacing
$\FFF(M)$ by a quasi-isomorphic chain complex will yield quasi-isomorphic
solution complexes $\Sol(M)$. (Recall that quasi-isomorphic chain complexes
should be  regarded as ``being the same''.) A homotopy pullback is a suitable deformation
(called a derived functor) of the ordinary pullback that does not suffer from
this problem. Consequently, our chain complex $\Sol(M)$ from 
Definition \ref{def:solutioncomplex} is invariant (up to quasi-isomorphisms)
under changing $\FFF(M)$ by quasi-isomorphisms. One should think
of our solution complex $\Sol(M)$ as enforcing the equation of motion
$P s_0=0$ in only a weak sense, i.e.\ ``up to homotopy''.
\end{rem}

\begin{propo}\label{propo:solutioncomplex}
A model for the solution complex $\Sol(M)$ 
from Definition \ref{def:solutioncomplex} is given by
\begin{flalign}\label{eqn:solutioncomplex}
\Sol(M) \,=\, \Big(
\xymatrix@C=2em{
\stackrel{(-2)}{\FFF_1(M)} 
& \ar[l]_-{Q^\ast} \stackrel{(-1)}{\FFF_0(M)} 
& \ar[l]_-{P }\stackrel{(0)}{\FFF_0(M)} 
& \ar[l]_-{Q} \stackrel{(1)}{\FFF_1(M)}
}
\Big)\quad.
\end{flalign}
\end{propo}
\begin{proof}
The homotopy pullback in \eqref{eqn:derivedlocus} can be computed by
using some basic model category technology, yielding the result in \eqref{eqn:solutioncomplex}. 
The proof for linear Yang-Mills theory in \cite[Proposition 3.21]{BSreview}
generalizes in a straightforward way to our present scenario and hence it will not be repeated.
\end{proof}

\begin{ex}\label{ex:solutionKG}
For the scalar field complex from Example \ref{ex:fieldKG}, 
we choose the massive Klein-Gordon operator $P= \square - m^2: \Omega^0(M)\to \Omega^0(M)$.
The action in \eqref{eqn:action} is then the usual Klein-Gordon action
\begin{flalign}
S(\Phi) = \frac{1}{2}\langle \Phi, \square\Phi - m^2\Phi \rangle=\frac{1}{2} \int_M
 \Big(\dd \Phi\wedge\ast \dd\Phi - m^2 \Phi^2\,\vol_M\Big)\quad.
\end{flalign}
The corresponding solution complex from Proposition \ref{propo:solutioncomplex}
explicitly reads as
\begin{flalign}
\Sol^{\KG}_{~}(M) \,=\, \Big(
\xymatrix@C=3em{
0
& \ar[l]_-{0} \stackrel{(-1)}{\Omega^0(M)}
& \ar[l]_-{\square-m^2 }\stackrel{(0)}{\Omega^0(M)}
& \ar[l]_-{0} 0
}
\Big)\quad.
\end{flalign}
The components of this complex admit a physical interpretation in terms of the
BRST/BV formalism:
\begin{itemize}
\item the fields in degree $0$ are the scalar fields $\Phi\in\Omega^0(M)$;
\item the fields in degree $-1$ are the antifields $\Phi^\ddagger \in \Omega^0(M)$.
\end{itemize}
Note that only the zeroth homology of $\Sol^{\KG}_{~}(M)$ is non-vanishing.
It is given by the ordinary solution space
$H_0(\Sol^{\KG}_{~}(M)) = \big\{\Phi \in \Omega^0(M)\,:\, \square\Phi - m^2\Phi =0\big\}$
of Klein-Gordon theory. It follows that $\Sol^{\KG}_{~}(M)$ is quasi-isomorphic
to its zeroth homology $H_0(\Sol^{\KG}_{~}(M)) $, regarded as a chain complex concentrated in degree $0$.
In other words, for Klein-Gordon theory on $M$ it does not make any difference if we work with the solution
complex $\Sol^{\KG}_{~}(M)$ or with the ordinary solution space $H_0(\Sol^{\KG}_{~}(M)) $.
\end{ex}

\begin{ex}\label{ex:solutionYM}
For the gauge field complex from Example \ref{ex:fieldYM}, 
we choose the linear Yang-Mills operator $P = \delta\dd: \Omega^1(M)\to \Omega^1(M)$.
The action in \eqref{eqn:action} is then the usual linear Yang-Mills action
\begin{flalign}
S(A) = \frac{1}{2}\langle A, \delta\dd A \rangle = \frac{1}{2}\langle \dd A, \dd A \rangle=
\frac{1}{2} \int_M F\wedge\ast F\quad,
\end{flalign}
with $F = \dd A\in\Omega^2(M)$ the field strength.
The corresponding solution complex from Proposition 
\ref{propo:solutioncomplex} explicitly reads as
\begin{flalign}
\Sol^{\YM}_{~}(M) \,=\, \Big(
\xymatrix@C=2em{
\stackrel{(-2)}{\Omega^0(M)} 
& \ar[l]_-{\delta} \stackrel{(-1)}{\Omega^1(M)} 
& \ar[l]_-{\delta\dd }\stackrel{(0)}{\Omega^1(M)} 
& \ar[l]_-{\dd} \stackrel{(1)}{\Omega^0(M)}
}
\Big)\quad.
\end{flalign}
The components of this complex admit a physical interpretation in terms of the
BRST/BV formalism:
\begin{itemize}
\item the fields in degree $0$ are the gauge fields $A\in\Omega^1(M)$;
\item the fields in degree $1$ are the ghost fields $c\in\Omega^0(M)$;
\item the fields in degrees $-1$ and $-2$ are the antifields
$A^\ddagger\in\Omega^1(M)$ and $c^\ddagger\in \Omega^0(M)$.
\end{itemize}
The homologies of $\Sol^{\YM}_{~}(M)$ can be computed explicitly and admit a physical interpretation.
\begin{itemize}
\item  $H_1(\Sol^{\YM}_{~}(M)) \cong H^0_\mathrm{dR}(M)$ is the zeroth de Rham cohomology 
of $M$. It describes those gauge transformations that act trivially on gauge fields, i.e.\ 
it encodes the extent to which the gauge group fails to act freely. This homology is never trivial, because 
the dimension of the vector space $H^0_\mathrm{dR}(M) \cong \bbR^{\pi_0(M)}$ is
given by the number of connected components of the manifold $M$.

\item $H_0(\Sol^{\YM}_{~}(M)) = \{A\in\Omega^1(M) : \delta\dd A =0\}\big/ \dd\Omega^0(M)$
is the usual vector space of gauge equivalence classes of linear Yang-Mills solutions.

\item $H_{-1}(\Sol^{\YM}_{~}(M)) \cong H^1_\delta(M)\cong H^{m-1}_{\mathrm{dR}}(M)$
is the first $\delta$-cohomology or equivalently the $m{-}1$-th de Rham cohomology of $M$.
It captures obstructions to solving the inhomogeneous linear Yang-Mills equation
$\delta\dd A = j$ with $j\in\Omega^1_\delta(M)$ a $\delta$-closed $1$-form, i.e.\ $\delta j=0$.
For the explicit computation of $H_{-1}(\Sol^{\YM}_{~}(M)) $
one uses standard techniques from the theory of normally hyperbolic operators \cite{BGP,Bar}
in order to prove that $\delta\dd A = j$ admits a solution $A$ if and only if $j=\delta \zeta$
is $\delta$-exact.

\item $H_{-2}(\Sol^{\YM}_{~}(M))\cong H^0_\delta(M) \cong H^m_{\mathrm{dR}}(M)\cong 0$
is the zeroth $\delta$-cohomology or equivalently the $m$-th de Rham cohomology of $M$.
This is trivial because every globally hyperbolic Lorentzian manifold is diffeomorphic
to a product manifold $M\cong \bbR\times \Sigma$.
\end{itemize}
We in particular observe that $\Sol^{\YM}_{~}(M)$ can not be quasi-isomorphic
to a chain complex concentrated in degree $0$, hence it contains more refined information than the
vector space of gauge equivalence classes of linear Yang-Mills solutions, i.e.\ the zeroth homology $H_0(\Sol^{\YM}_{~}(M))$.
It is the latter that is traditionally considered in the AQFT literature, see e.g.\ 
\cite{SDH,BDS14,BDHS14,FewsterLang,BeniniMaxwell,BSSdiffcoho}.
\end{ex}


\section{\label{sec:Poisson}Shifted and unshifted Poisson structures}
A general result of derived algebraic geometry \cite{DAG,DAG2,Pridham}
states that every derived critical locus comes endowed with a shifted symplectic 
structure and hence a shifted Poisson structure. Such shifted Poisson structures 
play a fundamental role in the factorization algebra approach to 
quantum field theory by Costello and Gwilliam \cite{CostelloGwilliam}. 
We explain below that the solution complex $\Sol(M)$ from Proposition \ref{propo:solutioncomplex} 
carries a natural shifted Poisson structure. For our two examples given by Klein-Gordon and 
linear Yang-Mills theory, we shall make the interesting observation that this shifted Poisson 
structure defines a trivial homology class, which crucially relies on our hypothesis that 
$M$ is a globally hyperbolic Lorentzian manifold. In these examples there exist two distinct 
types of chain homotopies (called retarded and advanced) that trivialize the shifted Poisson 
structure, which play an analogous role to the retarded and advanced Green's operators
in ordinary field theory, see e.g.\ \cite{BGP,BDH}.
Taking the difference between a compatible pair of retarded and advanced trivializations 
allows us to define an unshifted Poisson structure on $\Sol(M)$, which is the necessary 
ingredient for canonical commutation relations (CCR) quantization in Section \ref{sec:Quantization}.
\sk

Both the shifted and unshifted Poisson structures will be defined
on the smooth dual of the solution complex from Proposition 
\ref{propo:solutioncomplex}, which should be interpreted
as a chain complex of linear observables. 
\begin{defi}\label{def:linobscomplex}
The {\em complex of linear observables} for the solution complex $\Sol(M)$ 
from \eqref{eqn:solutioncomplex} is defined by
\begin{flalign}\label{eqn:LLLcomplex}
\LLL(M) \,:=\, \Big(
\xymatrix@C=2em{
\stackrel{(-1)}{\FFF_{1,\cc}(M)} 
& \ar[l]_-{-Q^\ast} \stackrel{(0)}{\FFF_{0,\cc}(M)} 
& \ar[l]_-{P }\stackrel{(1)}{\FFF_{0,\cc}(M)} 
& \ar[l]_-{-Q} \stackrel{(2)}{\FFF_{1,\cc}(M)}
}
\Big)\quad,
\end{flalign}
where the subscript $\cc$ denotes compactly supported sections.
The integration pairings \eqref{eqn:fiberpairingintegration} define
evaluation chain maps
\begin{subequations}\label{eqn:LLLcomplexevaluation}
\begin{flalign}
\langle \,\cdot\, , \,\cdot\, \rangle \,: \,\LLL(M)\otimes\Sol(M)~&\longrightarrow~\bbR \\
\langle \,\cdot\, , \,\cdot\, \rangle \,: \,\Sol(M)\otimes\LLL(M)~&\longrightarrow~\bbR
\end{flalign}
\end{subequations}
between linear observables and solutions.
\end{defi}

In order to define the shifted Poisson structure on $\Sol(M)$,
let us consider the $[1]$-shifting (see Section \ref{subsec:Ch}) 
of the solution complex $\Sol(M)$ in \eqref{eqn:solutioncomplex}, i.e.\
\begin{flalign}\label{eqn:Solshifted}
\Sol(M)[1]\,=\,\Big(
\xymatrix@C=2em{
\stackrel{(-1)}{\FFF_1(M)} 
& \ar[l]_-{-Q^\ast} \stackrel{(0)}{\FFF_0(M)} 
& \ar[l]_-{-P }\stackrel{(1)}{\FFF_0(M)} 
& \ar[l]_-{-Q} \stackrel{(2)}{\FFF_1(M)}
}
\Big)\quad,
\end{flalign}
and observe that the inclusion maps $\iota : \FFF_{n,\cc}(M) \to \FFF_n(M)$
of compactly supported sections define a chain map
\begin{flalign}\label{eqn:jmapping}
\parbox{0.5cm}{\xymatrix{
\LLL(M)\ar[d]_-{j}\\
\Sol(M)[1]
}
}~:=~~~~ \left(\parbox{2cm}{\xymatrix@C=2em{
\FFF_{1,\cc}(M) \ar[d]_-{\iota}
& \ar[l]_-{-Q^\ast} \FFF_{0,\cc}(M) \ar[d]_-{\iota}
& \ar[l]_-{P } \FFF_{0,\cc}(M)\ar[d]_-{-\iota}
& \ar[l]_-{-Q} \FFF_{1,\cc}(M) \ar[d]_-{-\iota}\\ 
\FFF_1(M)
& \ar[l]^-{-Q^\ast} \FFF_0(M)
& \ar[l]^-{-P } \FFF_0(M)
& \ar[l]^-{-Q} \FFF_1(M)
}
}\right)\quad.
\end{flalign}

\begin{defi}\label{def:shiftedPoisson}
The {\em shifted Poisson structure} on the solution complex
$\Sol(M)$ in \eqref{eqn:solutioncomplex} 
is the chain map $\Upsilon : \LLL(M)\otimes\LLL(M)\to \bbR[1]$ 
defined by the composition
\begin{flalign}\label{eqn:shiftedPoisson}
\xymatrix{
\ar[d]_-{\id\otimes j~~} \LLL(M)\otimes\LLL(M) \ar[rr]^-{\Upsilon} && \bbR[1]\\
\LLL(M)\otimes\bbR[1]\otimes \Sol(M) \ar[rr]_-{\gamma\otimes\id} && \bbR[1]\otimes
\LLL(M)\otimes \Sol(M) \ar[u]_-{\id\otimes \langle\,\cdot\,,\, \cdot\,\rangle }
}
\end{flalign}
where $\gamma$ is the symmetric braiding in $\Ch_\bbR$,
$\LLL(M)$ is the complex of linear observables \eqref{eqn:LLLcomplex} for $\Sol(M)$ 
and we implicitly used the isomorphism $\Sol(M)[1] \cong \bbR[1] \otimes \Sol(M)$, 
see Section \ref{subsec:Ch}. 
\end{defi}

\begin{rem}\label{rem:shiftingconventions}
In the terminology of the BRST/BV formalism, the shifted Poisson bracket is called the antibracket.
\end{rem}

As we explain in detail in the two subsections below, our examples given by Klein-Gordon 
and linear Yang-Mills theory on an oriented and time-oriented globally hyperbolic Lorentzian manifold $M$
have the interesting feature that the homology class $[j]=0\in H_0(\hom(\LLL(M),\Sol(M)[1]))$
of the chain map \eqref{eqn:jmapping} is trivial and as a consequence
the homology class $[\Upsilon] =0 \in H_0(\hom(\LLL(M)\otimes \LLL(M),\bbR[1]))$
of the shifted Poisson structure is trivial too. We shall obtain an interpretation of the 
trivializations of $j$ as analogs of the Green's operators in ordinary field theory. Before working
out the details for our two examples, we would like to introduce some general 
terminology and definitions that will be useful for this task. First, let us introduce
the past/future compact analog of the complex \eqref{eqn:LLLcomplex}, i.e.\
\begin{flalign}\label{eqn:LLLcomplexshiftpcfc}
\LLL_{\pc/ \fc}(M)\,:=\, 
\Big(
\xymatrix@C=2em{
\stackrel{(-1)}{\FFF_{1,\pc/\fc}(M)} 
& \ar[l]_-{-Q^\ast} \stackrel{(0)}{\FFF_{0,\pc/\fc}(M)} 
& \ar[l]_-{P }\stackrel{(1)}{\FFF_{0,\pc/\fc}(M)} 
& \ar[l]_-{-Q} \stackrel{(2)}{\FFF_{1,\pc/\fc}(M)}
}
\Big)\quad.
\end{flalign}
Observe that the chain map $j$ in \eqref{eqn:jmapping} factors 
through the canonical inclusions $\iota : \LLL(M)\to \LLL_{\pc/\fc}(M)$, i.e.\ 
we have a commutative triangle
\begin{flalign}\label{eqn:jfactorization}
\xymatrix{
\ar[rd]_-{\iota~~} \ar[rr]^-{j}\LLL(M) &  & \Sol(M)[1]\\
&\LLL_{\pc/\fc}(M) \ar[ru]_-{~~j_{\pc/\fc}^{}}&
}
\end{flalign}
where $j_{\pc/\fc}$ is the evident extension of the chain map \eqref{eqn:jmapping} 
to sections with past/future compact support.
\begin{defi}\label{def:retadvtrivialization}
A {\em retarded/advanced trivialization} is a contracting homotopy of
the chain complex $\LLL_{\pc/\fc}(M)$, i.e.\ a $1$-chain
$\Lambda^\pm \in \hom(\LLL_{\pc/\fc}(M),\LLL_{\pc/\fc}(M))_1$
such that $\id = \partial \Lambda^\pm$.
\end{defi}

The following are some simple properties of retarded/advanced trivializations.
\begin{lem}\label{lem:retadvtriv}
\begin{itemize}
\item[a)] If $\Lambda^\pm \in \hom(\LLL_{\pc/\fc}(M),\LLL_{\pc/\fc}(M))_1$ is a retarded/advanced trivialization,
then $j = \partial(j_{\pc/\fc}\,\Lambda^\pm\,\iota)$ and 
$\Upsilon = \partial\big((\id\otimes \langle\,\cdot\,,\,\cdot\,\rangle) \,(\gamma\otimes\id)\,
(\id\otimes (j_{\pc/\fc}\,\Lambda^{\pm}\,\iota))\big)$.
In particular, the homology classes $[j]=0$ and $[\Upsilon]=0$ are trivial.

\item[b)] If $\Lambda^\pm, \widetilde{\Lambda}^\pm \in \hom(\LLL_{\pc/\fc}(M),\LLL_{\pc/\fc}(M))_1$
are two retarded/advanced trivializations, then 
$\widetilde{\Lambda}^\pm - \Lambda^\pm = \partial \lambda^\pm$ 
for some $2$-chain $\lambda^\pm \in \hom(\LLL_{\pc/\fc}(M),\LLL_{\pc/\fc}(M))_2$.

\item[c)] If $\Lambda^\pm \in \hom(\LLL_{\pc/\fc}(M), \LLL_{\pc/\fc}(M))_1$ is a pair of retarded/advanced 
trivializations, then 
\begin{flalign}\label{eqn:causalpropagator}
\Lambda \,:=\, j_{\pc}\,\Lambda^{+}\,\iota - j_{\fc}\,\Lambda^{-} \,\iota 
\,\in\, \hom\big(\LLL(M),\Sol(M)[1]\big)_1 
\end{flalign}
is a $1$-cycle, i.e.\ $\partial \Lambda =0$. Via the chain isomorphism \eqref{eqn:shiftmappingcomplex},
this defines a chain map $\Lambda : \LLL(M)\to \Sol(M)$ to the unshifted solution complex.
\end{itemize}
\end{lem}
\begin{proof}
Items~a) and c) are straightforward checks. For item~b) we note that 
the homology of $\LLL_{\pc/\fc}(M)$ is trivial because $\Lambda^\pm$
is by definition a contracting homotopy of $\LLL_{\pc/\fc}(M)$. Because
all objects in $\Ch_\bbR$ are fibrant and cofibrant, the mapping complex functor
$\hom$ preserves quasi-isomorphisms, hence the homology of $\hom(\LLL_{\pc/\fc}(M),\LLL_{\pc/\fc}(M))$
is trivial too. Since $\partial(\widetilde{\Lambda}^\pm - \Lambda^\pm)=\id-\id=0$, it 
then follows that there exists a $2$-chain $\lambda^\pm \in \hom(\LLL_{\pc/\fc}(M),\LLL_{\pc/\fc}(M))_2$
such that $\widetilde{\Lambda}^\pm - \Lambda^\pm=\partial\lambda^\pm$.
\end{proof}

\begin{rem}\label{rem:contractiblechoices}
Lemma \ref{lem:retadvtriv} b) states that retarded/advanced trivializations
are unique up to homotopy, provided they exist. From the proof of the lemma we see that even more is true
and that such homotopies are unique up to higher homotopies. Indeed, if
$\lambda^\pm,\widetilde{\lambda}^\pm \in \hom(\LLL_{\pc/\fc}(M),\LLL_{\pc/\fc}(M))_2$
are $2$-chains such that $\partial\widetilde{\lambda}^\pm =
\widetilde{\Lambda}^\pm - \Lambda^\pm = \partial \lambda^\pm$, then $\widetilde{\lambda}^\pm-\lambda^\pm$
is a $2$-cycle and hence, because of acyclicity of the mapping complex $\hom(\LLL_{\pc/\fc}(M),\LLL_{\pc/\fc}(M))$,
there exists a $3$-chain $\zeta^\pm \in \hom(\LLL_{\pc/\fc}(M),\LLL_{\pc/\fc}(M))_3$ such that
$\widetilde{\lambda}^\pm-\lambda^\pm = \partial \zeta^\pm$. The same argument
applies to even higher homotopies, which implies that retarded/advanced trivializations
are unique up to contractible choices.
\end{rem}

\begin{defi}\label{def:compatiblepair}
A pair $\Lambda^\pm \in \hom(\LLL_{\pc/\fc}(M),\LLL_{\pc/\fc}(M))_1$ of retarded/advanced trivializations
is called {\em compatible} if the corresponding chain map $\Lambda : \LLL(M)\to \Sol(M)$
from \eqref{eqn:causalpropagator} satisfies the formal skew-adjointness property
\begin{flalign}
\xymatrix{
\ar[d]_-{-\Lambda\otimes \id}\LLL(M)\otimes \LLL(M) \ar[rr]^-{ \id\otimes\Lambda} 
&&  \LLL(M)\otimes\Sol(M)\ar[d]^-{\langle\,\cdot\,,\,\cdot\,\rangle}\\
\Sol(M)\otimes\LLL(M) \ar[rr]_-{\langle\,\cdot\,,\,\cdot\,\rangle}&&\bbR
}
\end{flalign}
with respect to the integration pairings \eqref{eqn:LLLcomplexevaluation}.
\end{defi}

\begin{defi}\label{def:unshiftedPoisson}
Suppose that $ \Lambda^\pm \in \hom(\LLL_{\pc/\fc}(M),\LLL_{\pc/\fc}(M))_1$ is a 
compatible pair of retarded/advanced trivializations. 
The corresponding {\em unshifted Poisson structure} on the solution complex $\Sol(M)$ 
in \eqref{eqn:solutioncomplex} is the chain map $\tau : \LLL(M)\otimes\LLL(M)\to \bbR$ defined by
the composition
\begin{flalign}\label{eqn:unshiftedPoisson}
\xymatrix{
\ar[dr]_-{\id\otimes \Lambda~~} \LLL(M)\otimes\LLL(M) \ar[rr]^-{\tau} && \bbR\\
&\LLL(M)\otimes\Sol(M)\ar[ur]_-{\langle\,\cdot\,,\, \cdot\,\rangle }&
}
\end{flalign}
where $\LLL(M)$ is the complex of linear observables \eqref{eqn:LLLcomplex} for $\Sol(M)$
and $\Lambda$ is given in \eqref{eqn:causalpropagator}.
\end{defi}

\begin{rem}
Because  $\Lambda^\pm \in \hom(\LLL_{\pc/\fc}(M),\LLL_{\pc/\fc}(M))_1$ is by hypothesis 
a compatible pair (see Definition \ref{def:compatiblepair}), it follows that the unshifted
Poisson structure from Definition \ref{def:unshiftedPoisson} is (graded) antisymmetric, i.e.\
$\tau \, \gamma = -\tau$ with $\gamma$ the symmetric braiding in $\Ch_\bbR$. Hence,
$\tau$ canonically defines a chain map (denoted with abuse of notation by the same symbol)
\begin{flalign}
\tau \,:\, \LLL(M)\wedge \LLL(M) ~\longrightarrow~\bbR
\end{flalign}
on the (graded) exterior product, or equivalently a $0$-cycle
$\tau \in \hom(\bigwedge^2\LLL(M),\bbR)_0$ of the corresponding mapping complex.
This perspective will be valuable below for studying homotopies between unshifted Poisson structures.
\end{rem}

\begin{cor}\label{cor:tauinv}
Suppose that $ \Lambda^\pm,\widetilde{\Lambda}^\pm \in \hom(\LLL_{\pc/\fc}(M),\LLL_{\pc/\fc}(M))_1$ are two  
compatible pairs of retarded/advanced trivializations and denote the corresponding
unshifted Poisson structures by $\tau, \widetilde{\tau}\in  \hom(\bigwedge^2\LLL(M),\bbR)_0$. 
Then there exists a $1$-chain $\rho \in  \hom(\bigwedge^2\LLL(M),\bbR)_1 $ such that
$\widetilde{\tau}-\tau = \partial\rho $. In particular, $[\tau] = [\widetilde{\tau}]$ 
define the same homology class in $H_0(\hom(\bigwedge^2\LLL(M),\bbR))$.
\end{cor}
\begin{proof}
By Lemma \ref{lem:retadvtriv} b), there exists $\lambda^\pm \in \hom(\LLL_{\pc/\fc}(M),\LLL_{\pc/\fc}(M))_2$ 
such that $\widetilde{\Lambda}^\pm - \Lambda^\pm = \partial \lambda^\pm$, hence
\begin{flalign}\label{eqn:tauhomotopyTMP}
\widetilde{\tau} - \tau = \partial\Big(\langle\,\cdot\,,\,\cdot\,\rangle\, \big(\id\otimes (j_\pc\,\lambda^+\,\iota - j_\fc \,\lambda^-\,\iota) \big)\Big)=: \partial\widetilde{\rho}\quad.
\end{flalign}
Consider the decomposition 
$\widetilde{\rho} = \widetilde{\rho}_a + \widetilde{\rho}_s = \tfrac{1}{2}\,\widetilde{\rho}\,(\id - \gamma) 
+\tfrac{1}{2} \, \widetilde{\rho} \, (\id + \gamma)$ of $\widetilde{\rho}$ into its (graded) 
antisymmetric and symmetric parts.
Because both $\widetilde{\tau}$ and $\tau$ are (graded) antisymmetric, 
taking the (graded) antisymmetrization of \eqref{eqn:tauhomotopyTMP} implies that
$\widetilde{\tau} - \tau = \partial \widetilde{\rho}_a$ with the (graded) antisymmetric $1$-chain
$\widetilde{\rho}_a\in  \hom(\bigwedge^2\LLL(M),\bbR)_1$.
\end{proof}

\begin{rem}
We would like to emphasize that our Definition \ref{def:unshiftedPoisson} of unshifted Poisson structures 
leaves one important question unanswered: Do compatible pairs of retarded/advanced trivializations 
exist? We do already know from Lemma \ref{lem:retadvtriv} b) that, provided they exist,
retarded/advanced trivializations are unique up to homotopy, and so are their associated
unshifted Poisson structures, see Corollary \ref{cor:tauinv}.
Note that such questions are analogs of existence and uniqueness for Green's operators 
in ordinary field theory. We shall now investigate these issues in detail for 
Klein-Gordon and linear Yang-Mills theory. This will in particular clarify the relationship 
between retarded/advanced trivializations  and retarded/advanced Green's operators.
\end{rem}

\subsection{\label{subsec:KG}Klein-Gordon theory}
Recall the Klein-Gordon solution complex $\Sol^\KG_{~}(M)$ from Example
\ref{ex:solutionKG}. The corresponding complex of linear observables from
Definition \ref{def:linobscomplex} then reads as
\begin{flalign}\label{eqn:LLLKG}
\LLL^{\KG}_{~}(M) \,=\, \Big(
\xymatrix@C=2.5em{
0 
& \ar[l]_-{0} \stackrel{(0)}{\Omega^0_\cc(M)} 
& \ar[l]_-{\square-m^2}\stackrel{(1)}{\Omega^0_\cc(M)} 
& \ar[l]_-{0} 0
}
\Big)\quad.
\end{flalign}
Elements $\varphi\in\LLL^{\KG}_{0}(M)=\Omega^0_\cc(M)$ 
in degree $0$ are interpreted as linear scalar field observables
and elements $\alpha \in \LLL^{\KG}_{1}(M)=\Omega^0_\cc(M)$ in degree $1$ as linear
antifield observables. The evaluation of these observables on $\Sol^{\KG}_{~}(M)$
is described by \eqref{eqn:LLLcomplexevaluation} and reads as
\begin{flalign}
\langle \varphi,\Phi\rangle = \int_M \varphi\, \Phi \,\vol_M\quad,\qquad
\langle \alpha,\Phi^\ddagger \rangle = \int_M \alpha\, \Phi^\ddagger \,\vol_M\quad,
\end{flalign}
for all $\Phi\in \Sol^\KG_{0}(M)=\Omega^0(M)$ and all $\Phi^\ddagger\in \Sol^\KG_{-1}(M)=\Omega^0(M)$.
Note that only the zeroth homology of $\LLL^{\KG}_{~}(M)$ is non-vanishing.
It is given by the ordinary vector space $H_0(\LLL^{\KG}_{~}(M)) = 
\Omega_\cc^0(M)\big/(\square-m^2)\Omega_\cc^0(M)$ of linear on-shell observables for
Klein-Gordon theory, see e.g.\ \cite{BDH}. It follows that $\LLL^{\KG}_{~}(M)$ is quasi-isomorphic
to its zeroth homology $H_0(\LLL^{\KG}_{~}(M)) $, regarded as a chain complex concentrated in degree $0$.
In other words, for Klein-Gordon theory on $M$ it does not make any difference if we work with the 
complex of linear observables $\LLL^{\KG}_{~}(M)$ or with the ordinary vector space $H_0(\LLL^{\KG}_{~}(M)) $
of linear on-shell observables.
\sk

The shifted Poisson structure $\Upsilon^{\KG}_{~} : \LLL^{\KG}_{~}(M)\otimes \LLL^{\KG}_{~}(M)\to\bbR[1]$ 
from Definition \ref{def:shiftedPoisson} describes the following 
pairing between scalar field  observables and antifield observables
\begin{flalign}
\Upsilon^{\KG}_{~}(\alpha,\varphi) \,=\, - \int_M \alpha\,\varphi \,\vol_M \,=\, \Upsilon^{\KG}_{~}(\varphi,\alpha)\quad, 
\end{flalign}
for all $\alpha \in \LLL^{\KG}_{1}(M)=\Omega^0_\cc(M)$ and $\varphi\in\LLL^{\KG}_{0}(M)=\Omega^0_\cc(M) $.
\sk

Our next aim is to classify all retarded/advanced trivializations 
in the sense of Definition \ref{def:retadvtrivialization} for Klein-Gordon theory.
Recalling the complex $\LLL^{\KG}(M)$ from \eqref{eqn:LLLKG}, a retarded/advanced 
trivialization $\Lambda^\pm$ may be visualized by the down-right pointing arrows in the diagram
\begin{flalign}
\xymatrix@C=4em{
0 \ar[d]_-{0}\ar[dr]^-{0}
& \ar[l]_-{0}  \Omega^0_{\pc/\fc}(M) \ar[d]_-{\id}\ar[dr]^-{\Lambda^\pm_0}
& \ar[l]_-{\square-m^2} \Omega^0_{\pc/\fc}(M)\ar[d]_-{\id} \ar[dr]^-{0}
& \ar[l]_-{0} 0 \ar[d]_-{0}\\ 
0
& \ar[l]^-{0} \Omega^0_{\pc/\fc}(M)
& \ar[l]^-{\square-m^2}\Omega^0_{\pc/\fc}(M)
& \ar[l]^-{0} 0
}
\end{flalign}
which in the present case are simply given by the data of a single 
linear map $\Lambda^\pm_0 : \Omega^0_{\pc/\fc}(M)\to \Omega^0_{\pc/\fc}(M)$.
The condition $\id = \partial\Lambda^\pm $ is equivalent to the two equalities
\begin{flalign}\label{eqn:tmpKGtriv}
(\square-m^2)\,\Lambda^\pm_0 \,=\, \id\quad,\qquad
\Lambda^\pm_0\,(\square -m^2) \,= \,\id\quad.
\end{flalign}
\begin{propo}\label{propo:KG}
For Klein-Gordon theory, there exists a unique retarded/advanced trivialization
$\Lambda^\pm \in \hom(\LLL^{\KG}_{\pc/\fc}(M),\LLL^{\KG}_{\pc/\fc}(M))_1$.
It is given by the unique (extended) retarded/advanced Green's operator
$\Lambda_0^\pm= G^\pm : \Omega^0_{\pc/\fc}(M)\to \Omega^0_{\pc/\fc}(M)$
for $\square-m^2$.
\end{propo}
\begin{proof}
Recall from Section \ref{subsec:Green} that the (extended) retarded/advanced Green's 
operator $G^\pm : \Omega^0_{\pc/\fc}(M)\to \Omega^0_{\pc/\fc}(M)$ for $\square-m^2$ satisfies 
\eqref{eqn:tmpKGtriv} and hence setting $\Lambda_0^\pm= G^\pm$ defines a retarded/advanced 
trivialization $\Lambda^\pm \in \hom(\LLL^{\KG}_{\pc/\fc}(M),\LLL^{\KG}_{\pc/\fc}(M))_1$. 
Uniqueness follows from Lemma \ref{lem:retadvtriv} b) and the fact that 
$\hom(\LLL^{\KG}_{\pc/\fc}(M),\LLL^{\KG}_{\pc/\fc}(M))_2=0$ is the zero vector space.
\end{proof}

Because of \eqref{eqn:Gskewadjoint}, the unique retarded and advanced trivializations 
obtained above form a compatible pair in the sense of Definition \ref{def:compatiblepair}. 
The corresponding unshifted Poisson structure from Definition
\ref{def:unshiftedPoisson} then reads as
\begin{flalign}\label{eqn:tauKG}
\tau^{\KG}(\varphi_1,\varphi_2) \, =\, \int_M \varphi_1\, G\varphi_2 \,\vol_M \quad,
\end{flalign}
for all $\varphi_1, \varphi_2 \in \LLL^{\KG}_{0}(M) = \Omega^0_\cc(M)$,
where $G := G^+ - G^-$ is the causal propagator for $\square - m^2$.
This is precisely the usual Poisson structure for Klein-Gordon theory, see e.g.\ \cite{BDH}.

\subsection{\label{subsec:YM}Linear Yang-Mills theory}
Recall the linear Yang-Mills solution complex $\Sol^\YM_{~}(M)$ from 
Example \ref{ex:solutionYM}. The corresponding complex of linear observables from Definition
\ref{def:linobscomplex} reads as
\begin{flalign}\label{eqn:LLLYM}
\LLL^{\YM}_{~}(M) \,=\, \Big(
\xymatrix@C=2.5em{
\stackrel{(-1)}{\Omega^0_\cc(M)} 
& \ar[l]_-{-\delta} \stackrel{(0)}{\Omega^1_\cc(M)} 
& \ar[l]_-{\delta\dd}\stackrel{(1)}{\Omega^1_\cc(M)} 
& \ar[l]_-{-\dd} \stackrel{(2)}{\Omega^0_\cc(M)}
}
\Big)\quad.
\end{flalign}
Elements $\varphi \in \LLL^{\YM}_{0}(M) = \Omega^1_\cc(M)$ in degree $0$ are interpreted 
as linear gauge field observables and elements $\chi \in \LLL^{\YM}_{-1}(M) = \Omega^0_\cc(M)$ 
in degree $-1$ as linear ghost field observables. Elements 
$\alpha \in \LLL^{\YM}_{1}(M) = \Omega^1_\cc(M)$ in degree $1$ 
and $\beta \in \LLL^{\YM}_{2}(M) = \Omega^0_\cc(M)$ in degree $2$ are interpreted as 
linear observables for the antifields $A^\ddagger$ and $c^\ddagger$. 
The evaluation of these observables on $\Sol^\YM_{~}(M)$ is described 
by \eqref{eqn:LLLcomplexevaluation} and reads as
\begin{subequations}
\begin{flalign}
\langle \varphi, A \rangle &= \int_M \varphi \wedge \ast A ~~\quad,\qquad
\langle \chi , c \rangle = \int_M \chi \, c\, \vol_M \quad, \qquad\\
\langle \alpha , A^\ddagger \rangle &=\int_M \alpha \wedge \ast A^\ddagger \quad,\qquad
\langle \beta, c^\ddagger \rangle = \int_M \beta \, c^\ddagger \,\vol_M \quad,
\end{flalign}
\end{subequations}
for all gauge fields $A \in \Sol^\YM_{0}(M)=\Omega^1(M)$, ghost fields 
$c \in \Sol^\YM_{1}(M)=\Omega^0(M)$ and antifields $A^\ddagger\in \Sol^\YM_{-1}(M)=\Omega^1(M)$ 
and $c^\ddagger \in \Sol^\YM_{-2}(M)=\Omega^0(M)$.
The homologies of  $\LLL^{\YM}_{~}(M)$ can be computed explicitly and admit a physical interpretation,
see also Example \ref{ex:solutionYM}.
\begin{itemize}
\item $H_{-1}(\LLL^{\YM}_{~}(M)) = H_{\cc,\delta}^0(M) \cong H_{\cc,\mathrm{dR}}^m(M)$ 
is by Poincar\'e duality the linear dual of the vector space $H_1(\Sol^{\YM}_{~}(M))
\cong H_{\mathrm{dR}}^0(M)$, i.e.\ it consists 
of linear observables testing those ghost fields that act trivially on gauge fields. 

\item $H_0(\LLL^{\YM}_{~}(M)) = \Omega_{\cc,\delta}^1(M) / \delta\dd \Omega_\cc^1(M)$ is the usual 
vector space of linear gauge-invariant on-shell observables, see e.g.\
\cite{SDH,BDS14,BDHS14,FewsterLang,BeniniMaxwell,BSSdiffcoho}.

\item $H_{1}(\LLL^{\YM}_{~}(M)) =\Omega^1_{\cc,\delta\dd} (M)/ \dd \Omega_\cc^0(M) 
\cong H_{\cc,\mathrm{dR}}^1(M)$ is by Poincar\'e duality the linear dual of the vector space 
$H_{-1}(\Sol^{\YM}_{~}(M)) \cong H_{\mathrm{dR}}^{m-1}(M)$, i.e.\ it consists of 
linear observables testing obstructions to solving the inhomogeneous linear 
Yang-Mills equation $\delta\dd A = j$ with $j\in\Omega^1_\delta(M)$.

\item $H_{2}(\LLL^{\YM}_{~}(M)) = H^0_{\cc,\mathrm{dR}}(M) \cong 0$, 
because $M\cong \bbR\times\Sigma$.
\end{itemize}

The shifted Poisson structure $\Upsilon^{\YM}_{~} : \LLL^{\YM}_{~}(M)\otimes \LLL^{\YM}_{~}(M)\to\bbR[1]$ 
from Definition \ref{def:shiftedPoisson} describes the following 
pairing between gauge or respectively ghost field observables and their 
corresponding antifield observables
\begin{subequations}
\begin{flalign}
\Upsilon^{\YM}_{~}(\alpha,\varphi) \,&=\, - \int_M \alpha \wedge \ast \varphi \,=\, \Upsilon^{\YM}_{~}(\varphi,\alpha)\quad,\qquad\\
\Upsilon^{\YM}_{~}(\beta,\chi) \,&=\, \int_M \beta \, \chi\,\vol_M \,=\, \Upsilon^{\YM}_{~}(\chi,\beta)\quad, 
\end{flalign}
\end{subequations}
for all $\varphi \in \LLL^{\YM}_{0}(M) = \Omega^1_\cc(M)$, 
$\alpha \in \LLL^{\YM}_{1}(M) = \Omega^1_\cc(M)$,
$\chi \in \LLL^{\YM}_{-1}(M) = \Omega^0_\cc(M)$ 
and $\beta \in \LLL^{\YM}_{2}(M) = \Omega^0_\cc(M)$.
\sk

We now construct a compatible pair of retarded/advanced trivializations
in the sense of Definitions \ref{def:retadvtrivialization} and \ref{def:compatiblepair} 
for linear Yang-Mills theory.
Recalling the complex $\LLL^{\YM}(M)$ from \eqref{eqn:LLLYM}, a retarded/advanced 
trivialization $\Lambda^\pm$ may be visualized by the down-right pointing arrows in the diagram
\begin{flalign}
\xymatrix@C=3.2em{
0 \ar[d]_-{0} \ar[dr]^-{0}
& \ar[l]_-{0} \Omega^0_{\pc/\fc}(M) \ar[d]_-{\id} \ar[dr]^-{\Lambda^\pm_{-1}}
& \ar[l]_-{-\delta} \Omega^1_{\pc/\fc}(M) \ar[d]_-{\id} \ar[dr]^-{\Lambda^\pm_0}
& \ar[l]_-{\delta\dd} \Omega^1_{\pc/\fc}(M) \ar[d]_-{\id} \ar[dr]^-{\Lambda^\pm_1}
& \ar[l]_-{-\dd} \Omega^0_{\pc/\fc}(M) \ar[d]_-{\id} \ar[dr]^-{0}
& \ar[l]_-{0} 0 \ar[d]_-{0}\\ 
0
& \ar[l]^-{0} \Omega^0_{\pc/\fc}(M)
& \ar[l]^-{-\delta} \Omega^1_{\pc/\fc}(M)
& \ar[l]^-{\delta\dd} \Omega^1_{\pc/\fc}(M)
& \ar[l]^-{-\dd} \Omega^0_{\pc/\fc}(M)
& \ar[l]^-{0} 0
}
\end{flalign}
which in the present case are three linear maps 
$\Lambda^\pm_{-1} : \Omega^0_{\pc/\fc}(M)\to \Omega^1_{\pc/\fc}(M)$, 
$\Lambda^\pm_0 : \Omega^1_{\pc/\fc}(M)\to \Omega^1_{\pc/\fc}(M)$ 
and $\Lambda^\pm_1 : \Omega^1_{\pc/\fc}(M)\to \Omega^0_{\pc/\fc}(M)$,
subject to the four identities  
\begin{flalign}\label{eqn:tmpYMtriv}
-\delta\, \Lambda^\pm_{-1} = \id \quad, \qquad 
\delta\dd\, \Lambda^\pm_0 - \Lambda^\pm_{-1}\, \delta = \id \quad, \qquad
\Lambda^\pm_0\, \delta\dd - \dd\, \Lambda^\pm_1 = \id \quad,\qquad
-\Lambda^\pm_1\, \dd = \id \quad.
\end{flalign}
\begin{propo}\label{propo:YMtrivializations}
Denote by $G^\pm: \Omega^1_{\pc/\fc}(M) \to \Omega^1_{\pc/\fc}(M)$ the
(extended) retarded/advanced Green's operators for the d'Alembert operator 
$\square: \Omega^1(M)\to\Omega^1(M)$ on $1$-forms.
The choices
\begin{flalign}\label{eqn:YMtriv}
\Lambda^\pm_{-1} \,=\, -G^\pm\, \dd \quad, \qquad 
\Lambda^\pm_0 \,=\, G^\pm \quad,\qquad
\Lambda^\pm_1 \,=\, -\delta\, G^\pm \quad
\end{flalign}
define a compatible pair of retarded/advanced trivializations for linear Yang-Mills theory.
\end{propo}
\begin{proof}
This follows immediately from the properties of Green's operators
stated in Section \ref{subsec:Green}, see in particular Example \ref{ex:forms}.
\end{proof}
\begin{rem}\label{rem:YMhomotopy}
In contrast to the example of Klein-Gordon theory from Section \ref{subsec:KG},
the retarded/advanced trivializations of linear Yang-Mills theory are not unique,
but only unique up to contractible choices, see Lemma \ref{lem:retadvtriv} and Remark \ref{rem:contractiblechoices}. 
Any other retarded/advanced trivialization $\widetilde{\Lambda}^\pm$
differs from our $\Lambda^\pm$ above by the differential of a $2$-chain 
$\lambda^\pm \in \hom(\LLL^{\YM}_{\pc/\fc}(M),\LLL^{\YM}_{\pc/\fc}(M))_2$. Explicitly,
the three non-zero components of $\widetilde{\Lambda}^\pm$ read as
\begin{flalign}
\begin{array}{rll}
\widetilde{\Lambda}^\pm_{-1} ~=
 & -G^\pm\, \dd 
  + \, \delta\dd\, \lambda^\pm_{-1} 
 & : \quad \Omega^0_{\pc/\fc}(M) ~ \longrightarrow~ \Omega^1_{\pc/\fc}(M) \quad, 
\\
\widetilde{\Lambda}^\pm_0 ~=
 & G^\pm 
  - \, \dd\, \lambda^\pm_0 + \lambda^\pm_{-1}\, \delta
 & : \quad \Omega^1_{\pc/\fc}(M)~ \longrightarrow~ \Omega^1_{\pc/\fc}(M) \quad, 
\\
\widetilde{\Lambda}^\pm_1 ~=
 & -\delta\, G^\pm 
  - \, \lambda^\pm_0\, \delta\dd
 & : \quad \Omega^1_{\pc/\fc}(M) ~ \longrightarrow~ \Omega^0_{\pc/\fc}(M) \quad,
\end{array}
\end{flalign}
where $\lambda^\pm_{-1}: \Omega^0_{\pc/\fc}(M) \to \Omega^1_{\pc/\fc}(M)$ 
and $\lambda^\pm_0: \Omega^1_{\pc/\fc}(M) \to \Omega^0_{\pc/\fc}(M)$ are 
the two non-zero components of the $2$-chain $\lambda^\pm$. 
\end{rem}

The unshifted Poisson structure (see Definition \ref{def:unshiftedPoisson}) that corresponds to our
compatible pair of retarded/advanced trivializations $\Lambda^\pm$ from Proposition 
\ref{propo:YMtrivializations} reads as
\begin{subequations}\label{eqn:tauYM}
\begin{flalign}
\tau^{\YM}(\varphi_1,\varphi_2) \,&=\, \int_M \varphi_1\wedge\ast G\varphi_2 = -\tau^{\YM}(\varphi_2,\varphi_1)\quad,\qquad\\
\tau^{\YM}(\alpha,\chi) \,&=\, - \int_M \alpha \wedge \ast G\dd\chi =  \tau^{\YM}(\chi,\alpha)\quad, 
\end{flalign}
\end{subequations}
for all $\varphi_1, \varphi_2 \in \LLL^{\YM}_0(M) = \Omega^1_\cc(M)$,
$\alpha \in \LLL^{\YM}_{1}(M) = \Omega^1_\cc(M)$ and $\chi\in \LLL^{\YM}_{-1}(M) = \Omega^0_\cc(M)$,
where $G := G^+ - G^-$ is the causal propagator for the d'Alembert operator $\square$ on $1$-forms.
Note that this Poisson structure acts non-trivially on pairs $(\varphi_1,\varphi_2)$ of linear gauge field observables
and also non-trivially on pairs $(\alpha,\chi)$ consisting of a linear antifield observable $\alpha$ and a linear
ghost field observable $\chi$. It extends to the richer level of chain complexes of linear observables 
$\LLL^{\YM}(M)$ the usual Poisson structure on linear gauge-invariant on-shell observables $H_0(\LLL^{\YM}(M))$,
see e.g.\ \cite{SDH,BDS14,BDHS14,FewsterLang,BeniniMaxwell,BSSdiffcoho}.
\sk

To conclude this section, we would like to emphasize that any other choice of a compatible pair
of retarded/advanced trivializations $\widetilde{\Lambda}^\pm$  (see Remark \ref{rem:YMhomotopy} 
for a concrete description) defines an unshifted Poisson structure $\widetilde{\tau} = \tau^{\YM} + \partial\rho$
that agrees with \eqref{eqn:tauYM} up to homotopy, see Corollary \ref{cor:tauinv}. We shall prove
in Proposition \ref{prop:CCRhomotopical} that the quantization of two homotopic Poisson structures
yields quasi-isomorphic observable algebras, i.e.\ the quasi-isomorphism type of the 
resulting quantum theory depends only on the {\em uniquely defined} homology classes $[\Lambda^\pm]
\in  H_1(\hom(\LLL^{\YM}_{\pc/\fc}(M),\LLL^{\YM}_{\pc/\fc}(M)))$.


\section{\label{sec:Quantization}Quantization}
The goal of this section is to develop a chain complex analog
of the usual canonical commutation relations (CCR) quantization
of vector spaces endowed with Poisson structures, see e.g.\ \cite{BDH}.
The input of our construction is a pair $(V,\tau)$ consisting of
a chain complex $V\in\Ch_\bbR$ and a chain map
$\tau : V\wedge V\to\bbR$. We shall call $(V,\tau)$ 
an {\em unshifted Poisson complex}. The output of our construction
is a differential graded unital and associative $\ast$-algebra
$\CCR(V,\tau)$ over the field of complex numbers $\bbC$
that implements the canonical commutation relations determined by $\tau$.
We shall investigate homotopical properties of this quantization prescription
and in particular prove that, up to quasi-isomorphism, the quantization $\CCR(V,\tau)$ 
does only depend on the quasi-isomorphism type of $(V,\tau)$ and on the
homology class $[\tau]\in H_0(\hom(\bigwedge^2 V,\bbR))$ of $\tau$.
In the context of our examples from Section \ref{sec:Poisson}, this means that
both Klein-Gordon theory and linear Yang-Mills theory can be consistently quantized
by our methods.
\sk

Let us now explain in some detail the CCR quantization $\CCR(V,\tau)$
of an unshifted Poisson complex $(V,\tau)$. We denote by $T_\bbC^\otimes V$
the free differential graded unital and associative $\ast$-algebra generated by $V\in\Ch_\bbR$.
Concretely, $T_\bbC^\otimes V$ is given by
\begin{flalign}
T_\bbC^\otimes V\,:=\, \bigoplus_{n=0}^\infty V_\bbC^{\otimes n}\quad,
\end{flalign}
where $V_\bbC := V\otimes\bbC\in\Ch_\bbC$ is the complexification of $V$,
together with the usual multiplication $\mu :  T_\bbC^\otimes V \otimes T_\bbC^\otimes V \to T_\bbC^\otimes V$
and unit $\eta : \bbC\to T_\bbC^\otimes V$ determined by 
$\mu((v_1\otimes\cdots\otimes v_n) \otimes (v_1^\prime\otimes\cdots\otimes v_m^\prime) ) = 
v_1\otimes\cdots\otimes v_n \otimes v_1^\prime\otimes\cdots\otimes v_m^\prime$,
for all $v_1,\dots,v_n,v_1^\prime,\dots,v_m^\prime\in V$, and $\oone := \eta(1) = 1\in V_\bbC^{\otimes 0}=\bbC$. 
The $\bbC$-antilinear $\ast$-involution is determined by $v^\ast = v$, for all $v\in V$.
The CCR quantization is defined as the quotient
\begin{flalign}
\CCR(V,\tau)\,:=\, T_\bbC^\otimes V\big/ \mathcal{I}_{(V,\tau)} 
\end{flalign}
by the two-sided differential graded $\ast$-ideal $\mathcal{I}_{(V,\tau)}\subseteq T_\bbC^\otimes V$
generated by the (graded) canonical commutation relations
\begin{flalign}\label{eqn:CCR}
v_1 \otimes v_2 - (-1)^{\vert v_1\vert \,\vert v_2\vert}\, v_2\otimes v_1 \,=\, i\,\tau(v_1,v_2)\,\oone\quad,
\end{flalign}
for all homogeneous elements $v_1,v_2\in V$ with degrees denoted by $\vert v_1\vert,\vert v_2\vert \in\bbZ$.
We note that CCR quantization is functorial 
\begin{flalign}\label{eqn:CCRfunctor}
\CCR \,:\, \PCh_\bbR~\longrightarrow~ \astdgAlg_\bbC
\end{flalign}
for the following  natural choices of categories:
\begin{itemize}
\item $\PCh_\bbR$ denotes the category of unshifted Poisson complexes, i.e.\ objects are
pairs $(V,\tau)$ consisting of a chain complex $V\in\Ch_\bbR$ and a chain map $\tau : V\wedge V\to\bbR$
and  morphisms $f: (V,\tau)\to (V^\prime,\tau^\prime)$ are chain maps $f : V\to V^\prime$ 
preserving the Poisson structures $\tau^\prime \, (f\wedge f) = \tau$.

\item $\astdgAlg_\bbC$ denotes the usual category of differential graded unital and associative $\ast$-algebras.
\end{itemize}
\begin{rem}\label{rem:CCR}
Every ordinary Poisson vector space $(V,\tau)$ defines an unshifted Poisson complex
whose underlying chain complex is concentrated in degree $0$.
In such cases our CCR quantization $\CCR(V,\tau)$ yields a differential graded
$\ast$-algebra concentrated in degree $0$, which coincides with the usual CCR algebra
from the non-homotopical framework, see e.g.\ \cite{BDH}.
\end{rem}

For our homotopical analysis of CCR quantization, we endow both $\PCh_\bbR$ 
and $\astdgAlg_\bbC$ with the structure of a {\em homotopical category} in the sense of
\cite{DHKS,Riehl}. This is a more flexible framework than model category theory,
which is very convenient for our purposes because $\PCh_\bbR$ is not a model category as 
it is not cocomplete. Similarly to model category theory, a homotopical category is a category 
with a choice of weak equivalences (containing all isomorphisms and satisfying the so-called 
$2$-of-$6$ property), however there is no need to introduce compatible classes of 
fibrations and cofibrations or to require the category to be bicomplete. In our context, we introduce the following
canonical homotopical category structures on $\PCh_\bbR$ and $\astdgAlg_\bbC$.
\begin{defi}\label{def:homotopicalCATS}
\begin{itemize}
\item[(i)] A morphism $f: (V,\tau)\to (V^\prime,\tau^\prime)$ in $\PCh_\bbR$
is a weak equivalence if its underlying chain map $f:V\to V^\prime$ is a quasi-isomorphism in $\Ch_\bbR$.

\item[(ii)] A morphism $\kappa : A\to A^\prime$ in $\astdgAlg_\bbC$ is a weak equivalence if its
underlying chain map  is a quasi-isomorphism in $\Ch_\bbC$.
\end{itemize}
\end{defi}

The next result shows that the CCR functor \eqref{eqn:CCRfunctor} has very pleasant
homotopical properties, which in particular ensure that our examples of linear gauge theories
from Section \ref{sec:Poisson} can be quantized consistently. The proof of the following
proposition is slightly technical and hence it will be carried out in detail in Appendix \ref{app:technical}.
\begin{propo}\label{prop:CCRhomotopical}
\begin{itemize}
\item[a)] The CCR functor \eqref{eqn:CCRfunctor} is a homotopical functor, i.e.\ it preserves
the weak equivalences introduced in Definition \ref{def:homotopicalCATS}.

\item[b)] Let $(V,\tau)\in \PCh_\bbR$ be an unshifted Poisson complex 
and $\rho\in \hom(\bigwedge^2V,\bbR)_1$ a $1$-chain.
Then there exists a zig-zag 
\begin{flalign}
\xymatrix@C=2em{
\CCR(V,\tau) ~&~\ar[l]_-{\sim}  A_{(V,\tau,\rho)} \ar[r]^-{\sim} ~&~ \CCR(V,\tau+\partial\rho)
}
\end{flalign}
of weak equivalences in $\astdgAlg_\bbC$.
\end{itemize}
\end{propo}

In our context of linear gauge theories from Section \ref{sec:Poisson}, 
we immediately obtain the following crucial result as a direct consequence 
of Proposition \ref{prop:CCRhomotopical} and Corollary \ref{cor:tauinv}.
\begin{cor}\label{cor:CCRindependenttriv}
Suppose that $ \Lambda^\pm,\widetilde{\Lambda}^\pm \in \hom(\LLL_{\pc/\fc}(M),\LLL_{\pc/\fc}(M))_1$ are two  
compatible pairs of retarded/advanced trivializations 
and denote the corresponding unshifted Poisson structures by
$\tau,\widetilde{\tau}: \LLL(M)\wedge \LLL(M) \to \bbR$.
Then the two CCR quantizations $\CCR(\LLL(M),\tau) \simeq \CCR(\LLL(M),\widetilde{\tau})$ 
are equivalent via a zig-zag of weak equivalences in $\astdgAlg_\bbC$. 
\end{cor}

\begin{ex}\label{ex:quantumKG}
Recall from Section \ref{subsec:KG} 
the unshifted Poisson complex $(\LLL^{\KG}_{~}(M),\tau^\KG)\in \PCh_\bbR$
for Klein-Gordon theory on an oriented and time-oriented globally hyperbolic Lorentzian manifold 
$M$, see in particular \eqref{eqn:LLLKG} and \eqref{eqn:tauKG}.
Observe that the quotient map $\LLL^{\KG}_{~}(M) \to H_0(\LLL^{\KG}_{~}(M)) = 
\Omega^0_\cc(M)/(\square-m^2)\Omega^0_\cc(M)$ to the vector space of linear on-shell
observables (regarded as a chain complex concentrated in degree $0$) is a quasi-isomorphism
and that the unshifted Poisson structure \eqref{eqn:tauKG} descends to the quotient
because of $G\,(\square-m^2)=0$. Hence, we obtain a weak equivalence
$(\LLL^{\KG}_{~}(M),\tau^\KG) \stackrel{\sim}{\to} (H_0(\LLL^{\KG}_{~}(M)),\tau^\KG)$ in $\PCh_\bbR$
from our original unshifted Poisson complex to an unshifted Poisson complex concentrated in degree $0$,
which is simply the ordinary Poisson vector space of linear on-shell observables. As a consequence
of Proposition \ref{prop:CCRhomotopical} a), it follows that our 
CCR quantization $\CCR(\LLL^{\KG}_{~}(M),\tau^{\KG})$
is weakly equivalent in $\astdgAlg_\bbC$ to the ordinary CCR quantization 
$\CCR(H_0(\LLL^{\KG}_{~}(M)),\tau^{\KG})$ of Klein-Gordon theory 
as a unital and associative $\ast$-algebra (regarded as a differential graded $\ast$-algebra 
concentrated in degree $0$).
\end{ex}

\begin{ex}\label{ex:quantumYM}
Recall from Section \ref{subsec:YM} 
the unshifted Poisson complex $(\LLL^{\YM}_{~}(M),\tau^\YM)\in \PCh_\bbR$
for linear Yang-Mills theory on an oriented and time-oriented globally hyperbolic Lorentzian manifold 
$M$, see in particular \eqref{eqn:LLLYM} and \eqref{eqn:tauYM}. As a consequence
of Corollary \ref{cor:CCRindependenttriv}, the CCR quantization $\CCR(\LLL^{\YM}_{~}(M),\tau^\YM)$
for our particular choice of retarded/advanced trivializations given in Proposition \ref{propo:YMtrivializations}
is equivalent via a zig-zag of weak equivalences in $\astdgAlg_\bbC$ to the CCR quantization 
$\CCR(\LLL^{\YM}_{~}(M),\widetilde{\tau} = \tau^\YM + \partial\rho)$ for any other choice, see also Remark
\ref{rem:YMhomotopy}. Thus, we obtain a consistent quantization prescription for linear Yang-Mills theory.
\sk

To be very explicit, let us also write out the (graded) commutation relations
of the generators of $\CCR(\LLL^{\YM}_{~}(M),\tau^\YM)$. We use a suggestive notation and
denote the smeared linear quantum observables for
gauge fields by $\widehat{A}(\varphi)$, for $\varphi\in\LLL^\YM_{0}(M)=\Omega^1_\cc(M)$,
the ones for ghost fields by $\widehat{c}(\chi)$, for $\chi\in\LLL^\YM_{-1}(M) = \Omega^0_\cc(M)$,
and the ones for antifields by $\widehat{A}^\ddagger(\alpha)$ and $\widehat{c}^\ddagger(\beta)$, 
for $\alpha\in\LLL^\YM_1(M)=\Omega^1_\cc(M)$ and $\beta \in \LLL^\YM_2(M) = \Omega^0_\cc(M)$.
Then  \eqref{eqn:tauYM} and \eqref{eqn:CCR} yield the following non-vanishing 
(graded) commutation relations
\begin{subequations}
\begin{flalign}
\big[\widehat{A}(\varphi_1),\widehat{A}(\varphi_2)\big] \,&=\, i \,\int_M \varphi_1\wedge\ast G\varphi_2~\oone 
\quad,\qquad\\
\big[\widehat{A}^\ddagger(\alpha),\widehat{c}(\chi)\big]\,&=\, -i\, \int_M \alpha \wedge \ast G\dd\chi ~\oone \, =\, \big[\widehat{c}(\chi),\widehat{A}^\ddagger(\alpha)\big] \quad,
\end{flalign}
\end{subequations}
for all $\varphi_1,\varphi_2\in\LLL^\YM_{0}(M)=\Omega^1_\cc(M)$, $\alpha\in\LLL^\YM_1(M)=\Omega^1_\cc(M)$
and $\chi\in\LLL^\YM_{-1}(M) = \Omega^0_\cc(M)$.
\end{ex}


\section{\label{sec:AQFT}Functoriality and homotopy AQFT axioms}
Our results and constructions in the previous sections considered a fixed 
oriented and time-oriented globally hyperbolic Lorentzian manifold $M$.
In order to obtain an algebraic quantum field theory (AQFT), in the original
sense of Haag and Kastler \cite{HK} or in the more modern
sense of Brunetti, Fredenhagen and Verch \cite{BFV}, we have to analyze
functoriality of our constructions with respect to a suitable
class of spacetime embeddings $f:M\to N$. The relevant categories 
are defined as follows:
\begin{itemize}
\item $\Loc$ denotes the category of oriented and time-oriented globally hyperbolic Lorentzian manifolds 
(of a fixed dimension $m\geq 2$) with morphisms $f: M\to N$ given by all orientation and
time-orientation preserving isometric embeddings whose image $f(M)\subseteq N$
is open and causally convex.

\item For any $\overline{M}\in\Loc$, we denote by $\Loc/\overline{M}$ the corresponding slice category. 
Its objects are all $\Loc$-morphisms $m : M\to \overline{M}$ with target $\overline{M}$ 
and its morphisms $f: (m :M\to \overline{M})\to (n: N\to \overline{M})$ 
are all commutative triangles
\begin{flalign}
\xymatrix{
\ar[dr]_-{m}M \ar[rr]^-{f} & & N \ar[dl]^-{n}\\
&\overline{M}&
}
\end{flalign}
in $\Loc$. Note that $\Loc/\overline{M}\simeq \mathbf{COpen}(\overline{M})$ is equivalent to the category of
all causally convex open subsets $U\subseteq \overline{M}$ with morphisms given by subset inclusion.
\end{itemize}

As observed in Examples \ref{ex:quantumKG} and \ref{ex:quantumYM},
our homotopy theoretical constructions naturally define {\em differential graded} $\ast$-algebras
of quantum observables for each spacetime $M$. As a consequence, the relevant
variants of AQFT to describe such models should take values in
the model category $\Ch_\bbC$ of chain complexes in contrast to the usual category
$\Vec_\bbC$ of vector spaces. Based on the recent operadic approach to AQFT
\cite{AQFToperad,involution}, such algebraic structures were systematically investigated in \cite{hAQFT}.
One of the outcomes of these studies is a concept of {\em homotopy 
AQFTs}, i.e.\ homotopy-coherent AQFTs that are obtained by a resolution of the relevant operad.
Since in the present paper our ground field $\bbC$ has characteristic $0$,
the strictification theorem of \cite{hAQFT} implies that every homotopy AQFT
can be strictified and hence all possible variants of homotopy AQFT are equivalent.
To describe our concrete examples in the present paper, it is sufficient
and very convenient to consider the following semi-strict model for homotopy AQFTs, 
where both functoriality and Einstein causality hold strictly, but the time-slice axiom
 is replaced by an appropriate homotopical analog.
\begin{defi}\label{def:hAQFT}
A (semi-strict) {\em homotopy AQFT on $\Loc$}
is a functor $\AAA: \Loc \to \astdgAlg_\bbC$ such that the following hold true: 
\begin{itemize}
\item[(i)] {\em Strict Einstein causality axiom:} For every pair $(f_1 : M_1 \to N, f_2 : M_2\to N)$ of $\Loc$-morphisms
with causally disjoint images, the chain map 
\begin{flalign}
\big[ \AAA(f_1)(-), \AAA(f_2)(-)\big] \,:\, \AAA(M_1) \otimes \AAA(M_2)~ \longrightarrow~ \AAA(N)
\end{flalign} 
is zero, where $[-,-] := \mu - \mu \,\gamma: \AAA(N) \otimes \AAA(N) \to \AAA(N)$ 
denotes the (graded) commutator in $\AAA(N)$.

\item[(ii)] {\em Homotopy time-slice axiom:} For every Cauchy morphism, i.e.\ a $\Loc$-morphism 
$f: M \to N$ such that
the image $f(M)\subseteq N$ contains a Cauchy surface of $N$,
the map $\AAA(f): \AAA(M) \overset{\sim}{\to} \AAA(N)$ is a weak equivalence in $\astdgAlg_\bbC$. 
\end{itemize}
For any $\overline{M} \in \Loc$, a {\em homotopy AQFT on $\overline{M}$} is a functor 
$\AAA: \Loc/\overline{M} \to \astdgAlg_\bbC$ on the slice category such that the evident
analogs of Einstein causality and time-slice hold true.
\end{defi}

\begin{rem}\label{rem:sliceforget}
Homotopy AQFTs on $\Loc$ are chain complex analogs of
theories in the sense of Brunetti, Fredenhagen and Verch  \cite{BFV},
while homotopy AQFTs on a fixed $\overline{M}\in\Loc$ are chain complex
analogs of theories in the sense of Haag and Kastler \cite{HK}.
Note that every homotopy AQFT $\AAA: \Loc \to \astdgAlg_\bbC$
on $\Loc$ defines a homotopy AQFT
$\AAA_{\overline{M}} := \AAA~\mathfrak{U}_{\overline{M}} : \Loc/\overline{M}\to \astdgAlg_\bbC$ 
on every $\overline{M}\in\Loc$ via precomposition with the forgetful functor $\mathfrak{U}_{\overline{M}} 
: \Loc/\overline{M}\to \Loc$. Explicitly, the latter is given on objects by $(m:M\to \overline{M})\mapsto M$
and on morphisms by $(f: (m:M\to \overline{M})\to (n : N\to \overline{M}))\mapsto (f :M\to N)$. 
\end{rem}

In the following let us assume that, in the context of Definition \ref{def:fieldcomplex},
the vector bundles $F_n\to M$ with fiber metrics $h_n$ and also the linear differential operator $Q$
are natural on $\Loc$. Using the associated pullbacks $f^\ast : \FFF_n(N)\to \FFF_n(M)$
of sections along $\Loc$-morphisms $f: M\to N$, this implies that
the assignment $M\mapsto \FFF(M)$ of field complexes \eqref{eqn:fieldcomplex} 
is contravariantly functorial, i.e.\
\begin{flalign}
\FFF\,:\,\Loc^{\op} ~\longrightarrow~\Ch_\bbR\quad.
\end{flalign}
Assuming further that the action \eqref{eqn:action} (or equivalently the linear differential
operator $P$) is natural implies that the 
assignment $M\mapsto \Sol(M)$ of solution complexes \eqref{eqn:solutioncomplex} 
is contravariantly functorial too, i.e.\
\begin{flalign}
\Sol\,:\,\Loc^{\op} ~\longrightarrow~\Ch_\bbR\quad.
\end{flalign}
Using also pushforwards $f_\ast : \FFF_{n,\cc}(M)\to \FFF_{n,\cc}(N)$
of compactly supported sections along $\Loc$-morphisms $f:M\to N$,
one observes that the assignment $M\mapsto \LLL(M)$ of complexes of linear observables
\eqref{eqn:LLLcomplex} is covariantly functorial, i.e.\
\begin{flalign}
\LLL\,:\,\Loc ~\longrightarrow~\Ch_\bbR\quad,
\end{flalign}
and that the integration pairings \eqref{eqn:LLLcomplexevaluation}
are natural in the sense that the diagram
\begin{flalign}
\xymatrix@C=4em{
\ar[d]_-{\id\otimes f^\ast } \LLL(M) \otimes \Sol(N) \ar[r]^-{f_\ast \otimes \id} & \LLL(N)\otimes \Sol(N)\ar[d]^-{\langle\,\cdot\, ,\,\cdot\,\rangle_N^{}}\\
\LLL(M)\otimes \Sol(M) \ar[r]_-{\langle\, \cdot\, ,\, \cdot\, \rangle_M^{}}&\bbR
}
\end{flalign}
commutes, for all $\Loc$-morphisms $f:M\to N$. Furthermore, one immediately observes that
the chain maps $j$ in \eqref{eqn:jmapping} are natural in the sense that the diagram
\begin{flalign}
\xymatrix@C=4em{
\ar[d]_-{j_M^{}} \LLL(M) \ar[r]^-{f_\ast} & \LLL(N)\ar[d]^-{j_N^{}}\\
\Sol(M)[1] &\ar[l]^-{f^\ast} \Sol(N)[1]
}
\end{flalign}
commutes, for all $\Loc$-morphisms $f:M\to N$, and that the 
shifted Poisson structures $\Upsilon$ in \eqref{eqn:shiftedPoisson}
are natural in the sense that the diagram
\begin{flalign}
\xymatrix@C=4em{
\ar[d]_-{f_\ast\otimes f_\ast} \LLL(M)\otimes \LLL(M) \ar[r]^-{\Upsilon_M^{}}& \bbR[1]\ar@{=}[d]\\
\LLL(N)\otimes \LLL(N)\ar[r]_-{\Upsilon_N^{}}& \bbR[1]
}
\end{flalign}
commutes, for all $\Loc$-morphisms $f:M\to N$.
\begin{rem}\label{rem:functorKGandYM}
Note that all our assumptions above on naturality of vector bundles and differential operators 
are satisfied for our examples of interest given by Klein-Gordon theory 
(see Examples \ref{ex:fieldKG} and \ref{ex:solutionKG}) and linear Yang-Mills theory 
(see Examples \ref{ex:fieldYM} and \ref{ex:solutionYM}).
In these examples $f^\ast$ is simply given by pullback of differential forms.
\end{rem}

\begin{rem}\label{rem:sliceforget2}
Using as in Remark \ref{rem:sliceforget} the forgetful functor
$\mathfrak{U}_{\overline{M}} : \Loc/\overline{M} \to \Loc$, all functors and natural transformations
on $\Loc$ that we introduced above can be restricted to the slice category
$\Loc/\overline{M}$, for each $\overline{M}\in\Loc$. This restricted data is sufficient 
when one attempts to construct only a homotopy AQFT on a fixed $\overline{M}\in\Loc$, 
in contrast to a homotopy AQFT on $\Loc$. 
\end{rem}

Our approach to construct unshifted Poisson structures (Definition \ref{def:unshiftedPoisson})
in terms of (compatible pairs of) retarded/advanced trivializations (Definition \ref{def:retadvtrivialization})
has to be supplemented by a suitable naturality axiom. Because the strength of 
our final result will depend on whether we work with $\Loc$ or a slice category $\Loc/\overline{M}$,
for some $\overline{M}\in\Loc$, we shall state our definitions and results below for both cases.
\begin{defi}\label{def:naturaltrivializations}
Let $\CC$ be either $\Loc$ or $\Loc/\overline{M}$, for any $\overline{M}\in\Loc$.
A {\em $\CC$-natural retarded/advanced trivialization} is a family 
\begin{subequations}
\begin{flalign}
\Lambda^\pm_{} \,:=\, \big\{\Lambda^\pm_M \in \hom(\LLL_{\pc/\fc}(M),\LLL_{\pc/\fc}(M))_1\big\}_{M\in\CC}^{}
\end{flalign}
of retarded/advanced trivializations for each $M\in\CC$, such that
\begin{flalign}
f^\ast\, (j_{\pc/\fc}\, \Lambda^\pm_N\,\iota)\, f_\ast=\, (j_{\pc/\fc}\, \Lambda^\pm_M\,\iota)\quad,
\end{flalign}
\end{subequations}
for all $\CC$-morphisms $f:M\to N$, see also \eqref{eqn:jfactorization}.
\end{defi}

\begin{rem}\label{rem:LocstrongerthanLocM}
Using as in Remark \ref{rem:sliceforget} the forgetful functor
$\mathfrak{U}_{\overline{M}} : \Loc/\overline{M} \to \Loc$, any $\Loc$-natural retarded/advanced trivialization
may be restricted to a $\Loc/\overline{M}$-natural retarded/advanced trivialization, for each $\overline{M}\in\Loc$.
Thus, it is in general harder to construct $\Loc$-natural retarded/advanced trivializations 
than $\Loc/\overline{M}$-natural ones.
\end{rem}

\begin{propo}\label{prop:naturalKGandYM}
\begin{itemize}
\item[a)] The unique retarded/advanced trivializations for Klein-Gordon theory
given in Proposition \ref{propo:KG} define a $\Loc$-natural retarded/advanced trivialization.
\item[b)] The retarded/advanced trivializations for linear Yang-Mills theory
given in Proposition \ref{propo:YMtrivializations} define a $\Loc$-natural retarded/advanced
trivialization.
\end{itemize}
\end{propo}
\begin{proof}
This is an immediate consequence of the standard result 
that the retarded/advanced Green's operators $G^\pm_M : \Omega_\cc^p(M)\to \Omega^p(M)$ 
for the d'Alembert operator $\square$ or the Klein-Gordon operator $\square-m^2$
satisfy the naturality condition $f^\ast\, G_N^\pm\, f_\ast = G_M^{\pm}$, for all
$\Loc$-morphisms $f:M\to N$. See \cite[Lemma 3.2]{BGproc} for a proof.
\end{proof}

\begin{defi}\label{def:naturalunshiftedPoisson}
Let $\CC$ be either $\Loc$ or $\Loc/\overline{M}$, for any $\overline{M}\in\Loc$.
A {\em $\CC$-natural unshifted Poisson structure} on the solution complex functor
$\Sol: \CC^{\op}\to \Ch_\bbR$ is a $0$-cycle $\tau\in  \hom({\textstyle\bigwedge^2\LLL},\bbR)_0$ 
in the chain complex
\begin{flalign}\label{eqn:homnaturalcomplex}
\hom\big({\textstyle\bigwedge^2\LLL},\bbR\big)\,:=\, \lim_{M\in\CC^\op} \hom\big({\textstyle\bigwedge^2\LLL(M)},\bbR\big)\,\in\, \Ch_\bbR \quad,
\end{flalign}
where $\LLL : \CC\to \Ch_\bbR$ is the functor assigning chain complexes of linear observables.
A {\em $\CC$-natural homotopy} between two $\CC$-natural unshifted Poisson structures
$\tau,\widetilde{\tau} \in \hom({\textstyle\bigwedge^2\LLL},\bbR)_0$
is a $1$-chain $\rho\in\hom({\textstyle\bigwedge^2\LLL},\bbR)_1 $, such that
$\widetilde{\tau}-\tau = \partial \rho$.
\end{defi}

\begin{rem}
We decided to state Definition \ref{def:naturalunshiftedPoisson} in a rather
abstract form because this will become useful later. From a more concrete perspective, 
the data of a $\CC$-natural unshifted Poisson structure $\tau \in \hom({\textstyle\bigwedge^2\LLL},\bbR)_0$ 
is given by a family
\begin{subequations}
\begin{flalign}
\big\{\tau_M^{} : \LLL(M)\wedge\LLL(M)\to\bbR\big\}_{M\in\CC}^{}
\end{flalign}
of chain maps, i.e.\ unshifted Poisson structures for each $M\in\CC$,
that satisfies the naturality condition
\begin{flalign}
\xymatrix@C=4em{
\ar[d]_-{f_\ast\wedge f_\ast}\LLL(M)\wedge\LLL(M)\ar[r]^-{\tau_M^{}} & \bbR\ar@{=}[d]\\
\LLL(N)\wedge\LLL(N) \ar[r]_-{\tau_N^{}} & \bbR
}
\end{flalign}
\end{subequations}
for all $\CC$-morphisms $f:M\to N$.  Similarly, a $\CC$-natural homotopy between
$\tau$ and $\widetilde{\tau}$  is a family $\{\rho_M^{} \in \hom(\bigwedge^2 \LLL(M),\bbR)_1\}_{M\in\CC}$
of $1$-chains, such that $\widetilde{\tau}_M^{} - \tau_M^{} = \partial\rho_M^{}$, for all $M\in\CC$,
and $\rho_N^{} \,(f_\ast\wedge f_\ast) = \rho_M^{}$, for all $\CC$-morphisms $f:M\to N$.
Similarly to Remark \ref{rem:LocstrongerthanLocM}, we note that $\Loc$-natural unshifted Poisson
structures and their homotopies are harder to construct than $\Loc/\overline{M}$-natural ones.
\end{rem}

\begin{lem}\label{lem:naturalproperties}
\begin{itemize}
\item[a)] Let $\CC$ be either $\Loc$ or $\Loc/\overline{M}$, for any $\overline{M}\in\Loc$,
and let $\Lambda^\pm$ be a $\CC$-natural compatible pair of
retarded/advanced trivializations. Then the component-wise
construction in Definition \ref{def:unshiftedPoisson}
defines a $\CC$-natural unshifted Poisson structure $\tau$.

\item[b)] Let $\CC= \Loc/\overline{M}$, for any $\overline{M}\in\Loc$.
Then the chain complex \eqref{eqn:homnaturalcomplex} is isomorphic
to the mapping complex $\hom(\bigwedge^2 \LLL(\overline{M}),\bbR)$ corresponding to $\overline{M}$.
As a consequence, every $\Loc/\overline{M}$-natural unshifted Poisson structure
$\tau$ is uniquely determined by an unshifted Poisson structure $\tau_{\overline{M}}^{}$
on $\overline{M}$ and every $\Loc/\overline{M}$-natural homotopy 
$\rho$ is uniquely determined by a homotopy $\rho_{\overline{M}}^{}$ on $\overline{M}$.

\item[c)] Suppose that $\Lambda^\pm$ and $\widetilde{\Lambda}^\pm$ are two $\Loc/\overline{M}$-natural 
compatible pairs of retarded/advanced trivializations, for any $\overline{M} \in \Loc$. 
Then the corresponding $\Loc/\overline{M}$-natural unshifted Poisson structures $\tau,\widetilde{\tau}$
from item~a) are homotopic, i.e.\ $\widetilde{\tau}-\tau = \partial\rho$ 
for some $\Loc/\overline{M}$-natural homotopy $\rho$.
\end{itemize}
\end{lem}
\begin{proof}
Item~a) is immediate because the definition of the unshifted Poisson structure in \eqref{eqn:unshiftedPoisson}
involves only natural maps. Item~c) follows from item~b) and Corollary \ref{cor:tauinv}. It thus remains
to prove item~b), which follows immediately from the fact that the slice category
$\Loc/\overline{M}$ has a terminal object $(\id: \overline{M}\to \overline{M})$, hence
$(\Loc/\overline{M})^\op$ has an initial object. The limit in \eqref{eqn:homnaturalcomplex} 
is then isomorphic to the chain complex $\hom(\bigwedge^2 \LLL(\overline{M}),\bbR)$ corresponding to this object. 
\end{proof}

\begin{rem}\label{rem:issuesLocnatural}
It is currently unclear to us if the analog of Lemma
\ref{lem:naturalproperties}~c) also holds true for the category $\Loc$.
Let us explain this issue in more detail. Suppose that
$ \Lambda^\pm,\widetilde{\Lambda}^\pm$ are two $\Loc$-natural 
compatible pairs of retarded/advanced trivializations and denote the
corresponding $\Loc$-natural unshifted Poisson structures by $\tau,\widetilde{\tau}$.
By Corollary \ref{cor:tauinv}, we obtain that for every $M\in\Loc$ there exists
a $1$-chain $\rho_M^{}\in\hom(\bigwedge^2\LLL(M),\bbR)_1$ such that
$\widetilde{\tau}_M^{} - \tau_M^{} = \partial\rho_M^{}$. However, it is unclear 
whether such homotopies can be chosen to be $\Loc$-natural as $\Loc$ has no
terminal object. (The terminal object in $\Loc/\overline{M}$ 
was crucial to prove Lemma \ref{lem:naturalproperties}~b) and hence c).)
As a consequence, it is currently unclear to us if the particular model for 
linear quantum Yang-Mills theory that we will construct below is, 
up to natural weak equivalences, the only possibility within our approach. 
In particular, we can not exclude the existence of 
a $\Loc$-natural compatible pair of retarded/advanced trivializations
different from the one in Proposition \ref{prop:naturalKGandYM}~b), 
that leads to a non-homotopic $\Loc$-natural unshifted Poisson structure
and hence potentially to a non-equivalent quantization.
\end{rem}

\begin{ex}\label{ex:KGnatural}
Let us apply our general results to Klein-Gordon theory, see 
Examples \ref{ex:fieldKG} and \ref{ex:solutionKG} as well as Section \ref{subsec:KG}.
The $\Loc$-natural compatible pair of retarded/advanced trivializations
from Proposition \ref{prop:naturalKGandYM} a) defines via Lemma \ref{lem:naturalproperties} a)
a $\Loc$-natural unshifted Poisson structure $\tau^{\KG}_{}$, whose components
$\tau^{\KG}_{M}$, for $M\in\Loc$, are given concretely by \eqref{eqn:tauKG}.
Due to the component-wise uniqueness result for retarded/advanced trivializations for Klein-Gordon theory
in Proposition \ref{propo:KG}, it follows that $\tau^{\KG}_{}$ is unique too. Hence, in the case of
Klein-Gordon theory we obtain stronger results than in the general Lemma \ref{lem:naturalproperties}.
\end{ex}

\begin{ex}\label{ex:YMnatural}
Let us now apply our general results to linear Yang-Mills theory, see 
Examples \ref{ex:fieldYM} and \ref{ex:solutionYM} as well as Section \ref{subsec:YM}.
The $\Loc$-natural compatible pair of retarded/advanced trivializations
from Proposition \ref{prop:naturalKGandYM} b) defines via Lemma \ref{lem:naturalproperties} a)
a $\Loc$-natural unshifted Poisson structure $\tau^{\YM}_{}$, whose components
$\tau^{\YM}_{M}$, for $M\in\Loc$, are given concretely by \eqref{eqn:tauYM}.
Unfortunately, as explained in Remark \ref{rem:issuesLocnatural}, we are currently unable
to exclude the existence of other $\Loc$-natural choices 
of compatible pairs of retarded/advanced trivializations that define 
non-homotopic $\Loc$-natural unshifted Poisson structures. The situation
gets much better when we work on a slice category $\Loc/\overline{M}$, for any $\overline{M}\in\Loc$.
In this case Proposition \ref{prop:naturalKGandYM} b) restricts to a 
$\Loc/\overline{M}$-natural compatible pair of retarded/advanced trivializations
and Lemma \ref{lem:naturalproperties} a) defines a $\Loc/\overline{M}$-natural 
unshifted Poisson structure $\tau^{\YM}_{}$. By Lemma \ref{lem:naturalproperties} c),
we know that any other choice of a $\Loc/\overline{M}$-natural compatible pair of 
retarded/advanced trivializations defines a homotopic $\Loc/\overline{M}$-natural 
unshifted Poisson structure. This means that, when restricted to
$\Loc/\overline{M}$, our constructions determine uniquely a homology class $[\tau^{\YM}_{}]$
in $H_0(\hom\big({\textstyle\bigwedge^2\LLL^\YM},\bbR\big))$.
\end{ex}

Let $\CC$ be either $\Loc$ or $\Loc/\overline{M}$, for any $\overline{M}\in\Loc$,
and suppose that we picked a $\CC$-natural unshifted Poisson structure $\tau$.
The assignment $M\mapsto (\LLL(M),\tau_M^{})$ defines a functor
\begin{flalign}
(\LLL,\tau)\,:\, \CC ~\longrightarrow ~\PCh_\bbR 
\end{flalign}
to the category of unshifted Poisson complexes, and post-composition with the CCR quantization functor
\eqref{eqn:CCRfunctor} defines a functor
\begin{flalign}\label{eqn:CCRpostcomposition}
\AAA := \CCR(\LLL,\tau)\,:\, \CC ~\longrightarrow~ \astdgAlg_\bbC 
\end{flalign}
to the category of differential graded $\ast$-algebras.
\sk

In order to analyze homotopical properties of this construction, we endow
both functor categories $\mathbf{Fun}(\CC, \PCh_\bbR )$ and $\mathbf{Fun}(\CC, \astdgAlg_\bbC )$
with the structure of a homotopical category \cite{DHKS,Riehl} in which weak equivalences
are so-called natural weak equivalences.
\begin{defi}\label{def:naturalweakequivalences}
Let $\CC$ be either $\Loc$ or $\Loc/\overline{M}$, for any $\overline{M}\in\Loc$.
\begin{itemize}
\item[(i)] A morphism in $\mathbf{Fun}(\CC, \PCh_\bbR )$ (i.e.\ a natural transformation)
is a {\em natural weak equivalence} if all its components are weak equivalences in 
$\PCh_\bbR$, see Definition \ref{def:homotopicalCATS}.

\item[(ii)] A morphism in
$\mathbf{Fun}(\CC, \astdgAlg_\bbC )$ (i.e.\ a natural transformation)
is a {\em natural weak equivalence} if all its components are weak equivalences in 
$\astdgAlg_\bbC$, see Definition \ref{def:homotopicalCATS}.

\item[(iii)] Let $\hAQFT(\CC)\subseteq \mathbf{Fun}(\CC, \astdgAlg_\bbC )$ denote the full subcategory
of functors satisfying the homotopy AQFT axioms from Definition \ref{def:hAQFT}.
A morphism in $\hAQFT(\CC)$ is a weak equivalence if and only if it is a natural 
weak equivalence in $\mathbf{Fun}(\CC, \astdgAlg_\bbC )$.
\end{itemize}
\end{defi}

\begin{rem}
Note that the weak equivalences in $\hAQFT(\CC)$ agree with those
considered in \cite{hAQFT}.
\end{rem}

The following result generalizes Proposition \ref{prop:CCRhomotopical} to the context
of functor categories.
\begin{propo}\label{propo:natzigzag}
Let $\CC$ be either $\Loc$ or $\Loc/\overline{M}$, for any $\overline{M}\in\Loc$.
\begin{itemize}
\item[a)] Post-composition with the CCR functor defines a homotopical functor
\begin{flalign}
\CCR\circ (-) \,:\,  \mathbf{Fun}(\CC, \PCh_\bbR )~\longrightarrow ~\mathbf{Fun}(\CC, \astdgAlg_\bbC )\quad.
\end{flalign}

\item[b)] Let $(V,\tau)\in \mathbf{Fun}(\CC, \PCh_\bbR )$ and $\rho \in \hom (\bigwedge^2 V,\bbR)_1$
a $\CC$-natural $1$-chain. Then there exists a zig-zag of natural weak equivalences in $\mathbf{Fun}(\CC, \astdgAlg_\bbC )$
connecting $\CCR(V,\tau)$ and $\CCR(V,\tau+\partial\rho)$.
\end{itemize}
\end{propo}
\begin{proof}
Item a)~is an immediate consequence of the component-wise definition 
of natural weak equivalences in Definition \ref{def:naturalweakequivalences}
and the result in Proposition \ref{prop:CCRhomotopical}  a)~that the CCR functor
is a homotopical functor.
\sk

Let us now focus on item b). By Proposition \ref{prop:CCRhomotopical} b) and the explicit construction in
Proposition \ref{propo:zigzag}, we obtain for each object $M \in \CC$ a zig-zag
\begin{flalign}\label{eqn:zigzagnaturalTMP}
\CCR(V(M),\tau_M^{}) ~\stackrel{\sim}{\longleftarrow}~ \QQQ_\mathrm{lin} (H_{(V(M),\tau_M^{},\rho_M^{})}) 
~\stackrel{\sim}{\longrightarrow}~ \CCR(V(M),\tau_M^{} + \partial\rho_M^{})
\end{flalign}
of weak equivalences in $\astdgAlg_\bbC$. From our construction of the object
$H_{(V(M),\tau_M^{},\rho_M^{})}$ in Proposition \ref{propo:zigzag}, 
one immediately observes that \eqref{eqn:zigzagnaturalTMP} 
are the components of a zig-zag of natural weak equivalences.
\end{proof}

Together with Lemma \ref{lem:naturalproperties}~c),
Proposition \ref{propo:natzigzag}~b) implies the following important result.
\begin{cor}\label{cor:naturalzigzag}
Fix any $\overline{M} \in \Loc$ and suppose that $\Lambda^\pm$ and $\widetilde{\Lambda}^\pm$ are 
two $\Loc/\overline{M}$-natural compatible pairs of retarded/advanced trivializations. Denote
the corresponding $\Loc/\overline{M}$-natural unshifted Poisson structures from 
Lemma \ref{lem:naturalproperties}~a) by $\tau$ and $\widetilde{\tau}$.
Then the two functors $\AAA := \CCR(\LLL,\tau)$ and $\widetilde\AAA := \CCR(\LLL,\widetilde\tau)$ 
are equivalent via a zig-zag of natural weak equivalences in $\mathbf{Fun}(\Loc/\overline{M}, \astdgAlg_\bbC )$. 
\end{cor}

The next lemma provides conditions on $(\LLL,\tau): \CC \to \PCh_\bbR$ 
which imply that $\AAA := \CCR(\LLL,\tau): \CC \to \astdgAlg_\bbC$ 
fulfills the homotopy AQFT axioms from Definition \ref{def:hAQFT}. 
\begin{lem}\label{lem:hAQFT}
Let $\CC$ be either $\Loc$ or $\Loc/\overline{M}$, for any $\overline{M} \in \Loc$,
and consider a functor $(\LLL,\tau): \CC \to \PCh_\bbR$.
\begin{itemize}
\item[a)]  If for every pair $(f_1 : M_1 \to N, f_2 : M_2\to N)$ of 
$\CC$-morphisms with causally disjoint images the chain map
\begin{flalign}
\tau\,(f_{1\,\ast}\otimes f_{2\,\ast})\,:\, \LLL(M_1)\otimes \LLL(M_2)~\longrightarrow~\LLL(N)
\end{flalign}
is zero, then the functor $\AAA := \CCR(\LLL,\tau): \CC \to \astdgAlg_\bbC$ satisfies Einstein causality.

\item[b)] If for every Cauchy morphism $f:M\to N$ the chain map
$f_\ast : \LLL(M)\to \LLL(N)$ is a quasi-isomorphism, then the functor 
$\AAA := \CCR(\LLL,\tau): \CC \to \astdgAlg_\bbC$ satisfies time-slice. 
\end{itemize}
\end{lem}
\begin{proof}
Item~a) is a direct consequence of the canonical commutation relations in \eqref{eqn:CCR}. 
Item~b) follows from the fact that $\CCR$ is a homotopical functor, 
see Proposition \ref{prop:CCRhomotopical}~a). 
\end{proof}

We are now ready to state and prove the main result of the present paper.
\begin{theo}\label{thm:KGandYMhAQFT}
\begin{itemize}
\item[a)] Let $\tau^\KG$ denote the unshifted Poisson structure defined 
by Lemma \ref{lem:naturalproperties}~a) from the unique $\Loc$-natural compatible pair 
of retarded/advanced trivializations for Klein-Gordon theory, see Proposition \ref{prop:naturalKGandYM}~a). 
Then the functor $\AAA^\KG := \CCR(\LLL^\KG,\tau^\KG): \Loc \to \astdgAlg_\bbC$ 
is a homotopy AQFT on $\Loc$, i.e.\ $\AAA^\KG \in \hAQFT(\Loc)$.

\item[b)] Let $\tau^\YM$ denote the unshifted Poisson structure defined 
by Lemma \ref{lem:naturalproperties}~a) from the $\Loc$-natural compatible pair 
of retarded/advanced trivializations for linear Yang-Mills theory, see Proposition \ref{prop:naturalKGandYM}~b). 
Then the functor $\AAA^\YM := \CCR(\LLL^\YM,\tau^\YM): \Loc \to \astdgAlg_\bbC$ 
is a homotopy AQFT on $\Loc$, i.e.\ $\AAA^\YM \in \hAQFT(\Loc)$. 
The restriction $\AAA^\YM_{\overline{M}}:= \AAA^{\YM}_{}~\mathfrak{U}_{\overline{M}}
\in \hAQFT(\Loc/\overline{M})$ given in Remark \ref{rem:sliceforget} 
defines a homotopy AQFT on each $\overline{M}\in\Loc$. Up to natural weak equivalence, 
these homotopy AQFTs on $\overline{M}$ do not depend on the choice of 
a $\Loc$-natural compatible pair of retarded/advanced trivializations.
\end{itemize}
\end{theo}
\begin{proof}
{\em Item a):} Example \ref{ex:KGnatural} defines a functor $(\LLL^\KG,\tau^\KG) : \Loc\to \PCh_\bbR$
and hence by post-composition with $\CCR$ a functor $\AAA^\KG := \CCR(\LLL^\KG,\tau^\KG): \Loc \to \astdgAlg_\bbC$.
It remains to prove that this functor satisfies the homotopy AQFT axioms from Definition \ref{def:hAQFT},
which we shall do by checking the sufficient conditions on $(\LLL^\KG,\tau^\KG)$ from Lemma \ref{lem:hAQFT}.
We deduce from the explicit expressions for $\tau^\KG$ in \eqref{eqn:tauKG}
and the support properties of retarded/advanced Green's operators (see Section \ref{subsec:Green})
that the hypothesis of Lemma \ref{lem:hAQFT}~a) is fulfilled, hence $\AAA^\KG$ satisfies Einstein causality.
\sk

In order to prove time-slice, recall from Example \ref{ex:quantumKG} that the quotient maps
$(\LLL^\KG(M),\tau^\KG_M) \to (H_0(\LLL^\KG(M)),\tau^\KG_M)$ are weak 
equivalences in $\PCh_\bbR$ for every $M \in \Loc$. This clearly defines
a natural weak equivalence $(\LLL^\KG,\tau^\KG) \to (H_0(\LLL^\KG),\tau^\KG)$
in $\mathbf{Fun}(\Loc,\PCh_{\bbR})$, hence we may equivalently
prove that $(H_0(\LLL^\KG),\tau^\KG)$ fulfills the hypothesis of Lemma \ref{lem:hAQFT}~b),
i.e.\ $H_0(f_\ast) : H_0(\LLL^\KG(M))\to H_0(\LLL^\KG(N))$ is a quasi-isomorphism
(i.e.\ an isomorphism because both chain complexes are concentrated in degree zero)
for every Cauchy morphism $f:M\to N$. Because $H_0(\LLL^\KG)$ describes
the usual vector spaces of linear on-shell observables for Klein-Gordon theory, 
this follows from standard results in the literature, see e.g.\  \cite[Theorem~3.3.1]{QFTmodels}.
This shows that $\AAA^\KG: \Loc \to \astdgAlg_\bbC$ satisfies the
homotopy AQFT axioms, hence it is a homotopy AQFT on $\Loc$.
\sk\sk

{\em Item b):} Example \ref{ex:YMnatural} defines a functor $(\LLL^\YM,\tau^\YM) : \Loc\to \PCh_\bbR$
and hence by post-composition with $\CCR$ a functor 
$\AAA^\YM := \CCR(\LLL^\YM,\tau^\YM): \Loc \to \astdgAlg_\bbC$.
We prove that this functor satisfies the homotopy AQFT axioms from Definition \ref{def:hAQFT}
by checking the sufficient conditions on $(\LLL^\YM,\tau^\YM)$ from Lemma \ref{lem:hAQFT}.
We deduce from the explicit expressions for $\tau^\YM$ in \eqref{eqn:tauYM}
and the support properties of retarded/advanced Green's operators (see Section \ref{subsec:Green})
that the hypothesis of Lemma \ref{lem:hAQFT}~a) is fulfilled, hence $\AAA^\YM$ satisfies Einstein causality.
\sk

Our next aim is to prove that the hypothesis of Lemma \ref{lem:hAQFT}~b) is fulfilled too, 
which would imply that $\AAA^\YM$ satisfies time-slice. 
Let $f:M\to N$ be any Cauchy morphism and consider the chain map $f_\ast: \LLL^\YM(M) \to \LLL^\YM(N)$,
where $\LLL^{\YM}$ is concretely given in \eqref{eqn:LLLYM}. We have to prove that 
the induced map $H_n(f_\ast): H_n(\LLL^\YM(M)) \to H_n(\LLL^\YM(N))$ in homology
is an isomorphism, for every $n\in\bbZ$. From the explicit computation of homologies
performed in Section \ref{subsec:YM}, we find that the only non-trivial homologies
are in degrees $n=-1,0,1$, hence we can restrict our attention to these cases.
The homologies in degrees $n=\pm1$ are compactly supported de Rham cohomologies,
hence $H_n(f_\ast)$ is an isomorphism in these degrees because of Poincar{\'e} duality,
homotopy invariance of de Rham cohomology and the fact that every 
Cauchy morphism $f: M \to N$ is in particular a homotopy equivalence. 
In degree $n=0$, the linear map $H_0(f_\ast): H_0(\LLL^\YM(M)) \to H_0(\LLL^\YM(N))$ 
is the usual push-forward along $f$ of linear gauge-invariant on-shell observables 
for linear Yang-Mills theory, which is known to be an isomorphism, see e.g.\ 
\cite{SDH,BDS14,BDHS14,FewsterLang,BeniniMaxwell,BSSdiffcoho}. 
\sk

Summing up, this shows that $\AAA^\YM: \Loc \to \astdgAlg_\bbC$ satisfies the
homotopy AQFT axioms, hence it is a homotopy AQFT on $\Loc$.
The statement about uniqueness (up to natural weak equivalences)
for the restricted linear Yang-Mills homotopy AQFTs $\AAA^{\YM}_{\overline{M}}\in \hAQFT(\Loc/\overline{M})$, 
for each $\overline{M}\in\Loc$, is a consequence of Corollary \ref{cor:naturalzigzag}.
\end{proof}

\begin{rem}\label{rem:hAQFTKG}
Note that our particular model in Theorem \ref{thm:KGandYMhAQFT}~a)
for Klein-Gordon theory as a homotopy AQFT on $\Loc$ is given by a functor
$\AAA^{\KG}$ that assigns a differential graded $\ast$-algebra to each $M\in\Loc$.
This is seemingly different to the usual description of Klein-Gordon theory as
a functor with values in ordinary $\ast$-algebras, see e.g.\ \cite{BDH}.
These two descriptions are however equivalent via a natural weak equivalence in
the homotopical category $\hAQFT(\Loc)$. Concretely, in the proof of
Theorem \ref{thm:KGandYMhAQFT}~a) we observed that there exists a
natural weak equivalence $(\LLL^\KG,\tau^\KG) \to (H_0(\LLL^\KG),\tau^\KG)$ 
between our unshifted Poisson complexes for Klein-Gordon theory and their
zeroth homologies, which are the structures of interest in the usual description of Klein-Gordon theory. 
Proposition \ref{propo:natzigzag}~a) then implies
that $\AAA^{\KG} = \CCR(\LLL^\KG,\tau^\KG) \simeq \CCR (H_0(\LLL^\KG),\tau^\KG)$
is a natural weak equivalence in $\hAQFT(\Loc)$, i.e.\ our description
of Klein-Gordon theory as a homotopy AQFT is equivalent to the usual one in e.g.\ \cite{BDH}.
\end{rem}

\begin{rem}\label{rem:hAQFTYM}
As we have already indicated in Remark \ref{rem:issuesLocnatural} and Example \ref{ex:YMnatural},
at the moment we cannot exclude the possibility that there exists another $\Loc$-natural compatible 
pair of advanced/retarded trivializations for linear Yang-Mills theory
that defines a non-homotopic $\Loc$-natural unshifted Poisson structure, and hence a potentially
non-equivalent homotopy AQFT $\widetilde{\AAA}^{\YM}\in\hAQFT(\Loc)$.
Note that potential differences would be very subtle
because, as a consequence of Theorem \ref{thm:KGandYMhAQFT}~b), the restrictions
$\widetilde{\AAA}^{\YM}_{\overline{M}}, \AAA^{\YM}_{\overline{M}}\in\hAQFT(\Loc/\overline{M})$
to every $\overline{M}\in\Loc$ are naturally weakly equivalent homotopy AQFTs on $\overline{M}$.
\sk

In contrast to the situation for Klein-Gordon theory explained in Remark \ref{rem:hAQFTKG},
our model in Theorem \ref{thm:KGandYMhAQFT}~b) for linear Yang-Mills theory as a homotopy
AQFT on $\Loc$ is {\em not} naturally weakly equivalent to existing models
in the literature that consider only gauge-invariant on-shell observables, 
see e.g.\ \cite{SDH,BDS14,BDHS14,FewsterLang,BeniniMaxwell,BSSdiffcoho}. This is because,
on a generic $M\in\Loc$, the complex of linear observables $\LLL^{\YM}(M)$ has 
non-trivial homology in degrees $n=-1,0,1$, while the usual models in the literature consider only
its zeroth homology. In the terminology of the BRST/BV formalism, one can say that our 
description of linear Yang-Mills theory as a homotopy AQFT $\AAA^{\YM}\in\hAQFT(\Loc)$ takes fully into account 
all ghost and antifield observables, while the traditional models consider only the $0$-truncation
of the antifield number $0$ sector of the theory. In particular, notice that the difference 
between $\AAA^\YM(M)$ and $\CCR(H_0(\LLL^\YM(M)),\tau^\YM_M)$, for a generic $M\in \Loc$,
is already visible on the level of the zeroth homology: 
The $\ast$-algebra $\CCR(H_0(\LLL^\YM(M)),\tau^\YM_M)$ 
is generated only by linear gauge-invariant on-shell observables, while 
the $\ast$-algebra $H_0(\AAA^\YM(M))$ contains also classes that are obtained by
multiplying in $\AAA^\YM(M)$ an equal number of ghost field and antifield linear observables. 
\end{rem}


\section*{Acknowledgments}
We would like to thank the anonymous referees for valuable 
comments that helped us to improve this manuscript.
We also would like to thank Chris Fewster, Owen Gwilliam, Igor Khavkine, 
Fran\c{c}ois Petit and Nic Teh for useful discussions and comments on this work.
The work of M.B.\ was supported by a research grant funded by 
the Deutsche Forschungsgemeinschaft (DFG, Germany). 
S.B.\ is supported by a PhD scholarship of the Royal Society (UK).
A.S.\ gratefully acknowledges the financial support of 
the Royal Society (UK) through a Royal Society University 
Research Fellowship, a Research Grant and an Enhancement Award.

\appendix

\section{\label{app:technical}Technical details for Proposition \ref{prop:CCRhomotopical}}
In this technical appendix we shall use quite freely the techniques and results developed in
\cite{BruinsmaSchenkel}. Let us first recall from \cite{BruinsmaSchenkel} that the CCR functor \eqref{eqn:CCRfunctor} admits a factorization
\begin{flalign}\label{eqn:CCRheisQlin}
\xymatrix@C=4em{
\ar@/^1.2pc/[rr]^-{\CCR} \PCh_\bbR \ar[r]_-{\mathfrak{Heis}}& \astdguLie_\bbC \ar[r]_-{\QQQ_\mathrm{lin}} & \astdgAlg_\bbC
}
\end{flalign}
through the homotopical category of  differential graded {\em unital Lie} $\ast$-algebras,
where $\mathfrak{Heis} : \PCh_\bbR\to  \astdguLie_\bbC$ is the Heisenberg Lie algebra functor
and $\QQQ_\mathrm{lin} : \astdguLie_\bbC \to \astdgAlg_\bbC $ is the unital universal 
enveloping algebra functor. Our strategy for proving  Proposition \ref{prop:CCRhomotopical}
is to prove the analogous statements for the Heisenberg Lie algebra functor
$\mathfrak{Heis} : \PCh_\bbR\to  \astdguLie_\bbC$. This will
imply our desired results, because of the following
\begin{lem}[\cite{BruinsmaSchenkel}]\label{lem:Qlinhomotopical}
$\QQQ_\mathrm{lin} : \astdguLie_\bbC \to \astdgAlg_\bbC $  is a homotopical functor.
\end{lem}

Let us recall that the Heisenberg Lie
algebra $\mathfrak{Heis}(V,\tau)\in \astdguLie_\bbC $ associated to an unshifted Poisson complex $(V,\tau)\in \PCh_\bbR$
is given by the chain complex
\begin{flalign}
\mathfrak{Heis}(V,\tau) \,:=\, V_\bbC \oplus \bbC\quad,
\end{flalign}
together with the Lie bracket 
$[-,-] : (V_\bbC \oplus \bbC)\otimes (V_\bbC \oplus \bbC)\to V_\bbC \oplus \bbC $ determined by
\begin{flalign}
\big[v_1\oplus c_1 , v_2\oplus c_2 \big] \,=\, 0\oplus i\, \tau(v_1,v_2)\quad, 
\end{flalign}
for all $v_1,v_2\in V$ and $c_1,c_2\in\bbC$, and unit $\eta : \bbC\to V_\bbC \oplus \bbC$ given
by $\oone:= \eta(1) = 0\oplus 1$. 
The $\ast$-involution on $V_\bbC \oplus \bbC$ is determined by $v^\ast = v$, for all $v\in V$,
and complex conjugation on $\bbC$.
To a morphism $f: (V,\tau)\to (V^\prime,\tau^\prime)$
in $\PCh_\bbR$ it assigns the $\astdguLie_\bbC$-morphism $\mathfrak{Heis}(f) : 
\mathfrak{Heis}(V,\tau)\to \mathfrak{Heis}(V^\prime,\tau^\prime)$ determined by 
$f_\bbC\oplus\id : V_\bbC \oplus \bbC\to V_\bbC^\prime \oplus \bbC$, where
$f_\bbC := f\otimes\id : V\otimes \bbC\to V^\prime \otimes \bbC$ denotes the complexification
of the chain map $f$. 
\sk

The following lemma follows  directly from the definitions.
\begin{lem}
$\mathfrak{Heis} : \PCh_\bbR\to  \astdguLie_\bbC$ is a homotopical functor.
Together with Lemma \ref{lem:Qlinhomotopical}, this proves Proposition \ref{prop:CCRhomotopical} a).
\end{lem}

The main technical result of this appendix is 
\begin{propo}\label{propo:zigzag}
Let $(V,\tau)\in \PCh_\bbR$ be an unshifted Poisson complex 
and $\rho\in \hom(\bigwedge^2V,\bbR)_1$ a $1$-chain.
Then there exists a zig-zag 
$\mathfrak{Heis}(V,\tau)   \stackrel{\sim}{\leftarrow} H_{(V,\tau,\rho)} \stackrel{\sim}{\rightarrow}
\mathfrak{Heis}(V,\tau+\partial\rho)$
of weak equivalences in $\astdguLie_\bbC$.
Together with Lemma \ref{lem:Qlinhomotopical}, this proves Proposition \ref{prop:CCRhomotopical} b).
\end{propo}
\begin{proof}
We construct an explicit object $H_{(V,\tau,\rho)} \in \astdguLie_\bbC$ that allows us to
exhibit the desired zig-zag of weak equivalences. Let us introduce the acyclic chain complex
\begin{flalign}
D \,:=\,  \Big(
\xymatrix@C=2.5em{
\stackrel{(-1)}{\bbC} &\ar[l]_-{\id} \stackrel{(0)}{\bbC} 
}
\Big)\,\in\,\Ch_\bbC
\end{flalign}
and the notations $x:= 1\in D_0$ and $y:= \dd x = 1\in D_{-1}$.
We define $H_{(V,\tau,\rho)} \in \astdguLie_\bbC$  by
\begin{flalign}
H_{(V,\tau,\rho)}\,:=\, V_\bbC \oplus D \oplus \bbC\quad,
\end{flalign}
together with the unit $\oone := 0\oplus0\oplus 1$ and the Lie bracket
\begin{flalign}
\big[v_1\oplus \alpha_1 \oplus c_1 , v_2\oplus\alpha_2\oplus c_2\big]
\,:= \, 0 \oplus\big(i\, \partial \rho(v_1,v_2)\, x + i\, \rho(v_1,v_2)\,y\big) \oplus i\,\tau(v_1,v_2)\quad.
\end{flalign}
For any real number $s\in\bbR$, we let $\mathcal{I}_s \subseteq H_{(V,\tau,\rho)}$
be the differential graded unital Lie $\ast$-algebra ideal generated by the two relations
\begin{flalign}\label{eqn:LieidealTMP}
0\oplus x\oplus 0 \,=\, 0 \oplus 0\oplus s\quad,\qquad
0\oplus y\oplus 0\,=\, 0\quad.
\end{flalign}
Note that the corresponding quotient
\begin{flalign}
 H_{(V,\tau,\rho)}\big/ \mathcal{I}_s  \,\cong\, \mathfrak{Heis}(V , \tau + s\,\partial \rho)
\end{flalign}
is isomorphic to the Heisenberg Lie algebra of $(V , \tau + s\,\partial \rho)\in\PCh_\bbR$. 
We still have to show that the quotient map
\begin{flalign}\label{eqn:pistmp}
\pi_s\,:\, H_{(V,\tau,\rho)}~\longrightarrow~ \mathfrak{Heis}(V , \tau + s\,\partial \rho)
\end{flalign}
is a weak equivalence in $\astdguLie_\bbC$. From the explicit
form of the relations in \eqref{eqn:LieidealTMP}, we observe that $\pi_s = \id_V \oplus q_s$
with $q_s : D\oplus \bbC \to \bbC$ given by $q_s : (c_1 x + c_2 \,y)\oplus c_3\mapsto
s\, c_1 + c_3$. This is clearly a quasi-isomorphism in $\Ch_\bbC$, hence 
\eqref{eqn:pistmp} is a weak equivalence for any $s\in\bbR$. The desired zig-zag follows
by taking $s=0$ and $s=1$.
\end{proof}


\end{document}